\documentclass[12pt]{amsart}
\usepackage{a4}
\usepackage{graphicx}
\usepackage{amsmath,amssymb,amsthm, color, fancyvrb, array}

\theoremstyle{plain}
\newtheorem{thm}{Theorem}[section]
\newtheorem{lem}[thm]{Lemma}
\newtheorem{cor}[thm]{Corollary}
\newtheorem{prop}[thm]{Proposition}
\theoremstyle{definition}
\newtheorem{defn}[thm]{Definition}

\newcommand{\R}{{\mathbb R}}

\newcommand{\tr}{\mathrm{tr}}
\newcommand{\di}{\mathrm{div}}

\newcounter{mnotecount}[section]
\let\oldmarginpar\marginpar
\renewcommand\marginpar[1]{\-\oldmarginpar[\raggedleft\footnotesize #1]%
{\raggedright\footnotesize #1}}

\title[Existence and Blowup]{Existence and Blowup Results for Asymptotically Euclidean Initial Data Sets Generated by the Conformal Method}

\author[J. Dilts and J. Isenberg]{James Dilts 
\\ University of California (San Diego)\\
\\
James Isenberg
\\ University of Oregon}

\date{\today}

\begin{document}
\maketitle
\begin{abstract} For each set of (freely chosen) seed data, the conformal method reduces the Einstein constraint equations to a system of elliptic equations, the conformal constraint equations. We prove an admissibility criterion, based on a (conformal) prescribed scalar curvature problem, which provides a necessary condition on the seed data for the conformal constraint equations to (possibly) admit a solution.  We then consider sets of asymptotically Euclidean (AE) seed data for which solutions of the conformal constraint equations exist, and examine the blowup properties of these solutions as the seed data sets approach sets for which no solutions exist. We also prove that there are AE seed data sets which include a Yamabe nonpositive metric and lead to solutions of the conformal constraints. These data sets  allow the mean curvature function to have zeroes.

\end{abstract}

\section{Introduction}
\label{Intro}

To construct a spacetime solution of the Einstein gravitational field equations, the first step is to find an initial data set which satisfies the Einstein constraint equations. That is, for a fixed manifold $\Sigma^n$, one seeks a Riemannian metric $\gamma$, a symmetric tensor field $K$, plus non-gravitational fields $\psi$ such that the Einstein constraint equations
\begin{align}
\label{HamConstr}
R_{\gamma} + (\tr_{\gamma} K)^2 - |K|_{\gamma}^2 &= 16 \pi \rho(\psi, \gamma) \\
\label{MomConstr}
\nabla^i K_{ij} - \nabla_j(\tr_{\gamma} K) &=8 \pi J(\psi, \gamma)
\end{align}
are satisfied everywhere on $\Sigma^n$. Here $R_{\gamma}$ is the scalar curvature of the metric, $\rho(\psi, \gamma)$ is the energy density of the non-gravitational fields, and $J(\psi, \gamma)$ is the momentum density for these fields. As an example, if the non-gravitational fields of interest are the electromagnetic vector fields $B$ and $E$ of the Einstein-Maxwell theory in 3+1 dimensions, then we have $\rho= \frac{1}{2}(|E|_{\gamma}^2 +|B|_{\gamma}^2)$ and $J =E\times_{\gamma} B$. For a perfect fluid, we may treat $\rho$ and $J$ themselves  as the
non-gravitational field initial data\footnote{A justification for this way of specifying initial data for perfect fluids is presented in  Section 4.1 of \cite{IMaxP05}.} (so long as the algebraic constraint $\rho^2\ge |J|^2_{\gamma}$ is satisfied).

The conformal method (along with the closely related conformal thin sandwich method) has proven to be a very effective procedure for producing as well as studying initial data sets $(\Sigma^n; \gamma, K, \psi)$ which satisfy the constraints \eqref{HamConstr}-\eqref{MomConstr}. It does this by splitting the initial data into two sets of fields: the freely chosen \emph{seed data}, and the \emph{determined data}. The idea is that, for a specified set of seed data, the constraint equations become a determined system to be solved  for the determined data.\footnote{As a PDE system to be solved for $(\Sigma^n; \gamma, K, \psi)$, the constraints form an underdetermined system.}

More specifically (working here with the Einstein-perfect fluid theory), the seed data set consists of a Riemannian metric $\lambda$,  a positive lapse function $N$, a symmetric tensor field $\sigma$ which is divergence-free and trace-free with respect to $\lambda$, a scalar function $\tau$, a nonnegative scalar function $r$, and a vector field $j$ such that $r^2\ge |j|^2_{\lambda}$. The determined data consists of a positive scalar field $\varphi$ and a vector field $W$. For a chosen set of seed data\footnote{While the divergence-free and trace-free conditions on $\sigma$ appear to result in  these fields not being freely specifiable, in fact by solving certain linear algebraic and linear PDE systems one readily projects out the divergence-free and trace-free pieces of any arbitrarily-chosen symmetric tensor field.}  $(\Sigma^n; \lambda, N, \sigma, \tau, r, j)$, one obtains $\varphi$ and $W$ by solving the \emph{conformal constraint equations}, which take the form 
\begin{equation}
\label{origLich}
\alpha_n \Delta_{\lambda} \varphi = R_{\lambda} \varphi +\kappa_n \tau^2 \varphi^{q_n -1} - \left|\sigma +\frac{1}{2N}L_{\lambda}W\right|_{\lambda}^2 \varphi^{-q_n -1} - r \varphi ^ {-\frac{q_n }{2}},
\end{equation}
\begin{equation}
\label{origVect}
\di_{\lambda}\Big( \frac1{2N} L_{\lambda}W\Big) = \kappa_n  \varphi^{q_n } d\tau + j.
\end{equation}
Here $\Delta_{\lambda}$ is the Laplacian (with negative eigenvalues)  with respect to the metric $\lambda$, $R_{\lambda}$ is its scalar curvature, $L_{\lambda}$ is the conformal Killing operator, defined by
\begin{equation}
L_{\lambda}W_{ij} := \nabla_i W_j + \nabla_j W_i - \frac{2}{n} \lambda_{ij} \nabla^k W_k
\end{equation}
for the $\lambda$-compatible covariant derivative $\nabla$, $\di_{\lambda}$ is the corresponding divergence operator, and we use the dimensional constants $\alpha_n := \frac{4(n-1)}{n-2}$, $\kappa_n  := \frac{n-1}{n}$, and $q_n := \frac{2n}{n-2}$. If indeed one does obtain a solution $(\varphi, W)$ to \eqref{origLich}-\eqref{origVect} for the given set of seed data, then the initial data set
\begin{align}
\gamma &=\varphi^{q_n -2} \lambda,\\
\label{KRecon}
K_{ij}&=\varphi^{-2} \left(\sigma_{ij} + \frac{1}{2N}L_{\lambda}W_{ij}\right)+ \frac{\tau}{n} \varphi^{q_n -2} \lambda_{ij}\\
\rho&=\varphi^{-\frac{3}{2} q_n +1} r/16\pi \\
J&= \varphi ^{-q_n } j/8\pi
\end{align}
is a solution of the constraint equations \eqref{HamConstr}-\eqref{MomConstr} on $\Sigma^n$. Note that if the seed data satisfies the inequality $r^2\ge |j|^2_{\lambda}$, it follows that the initial data satisfies the inequality $\rho^2\ge |J|^2_{\gamma}$. Here, we are using the Hamiltonian conformal thin sandwich (CTS-H) approach, as described in \cite{Maxwell14}. This adds the extra function $N$ to the seed data. The traditional conformal method simply
sets $N=1/2$. While the proofs in this paper can be carried out using the traditional method\footnote{In view of the equivalence of the two methods, as proven in \cite{Maxwell14}, this is no surprise.}, many of the calculations are simpler if one uses the CTS-H method since the CTS-H conformal constraint equations \eqref{origLich}-\eqref{origVect} are conformally invariant  \footnote{Explicitly, conformal covariance  takes the following form: $(\varphi, W)$ solve
the conformal constraint equations \eqref{origLich}-\eqref{origVect} with seed data $(\Sigma^n; \lambda, N, \sigma, \tau, r, j)$ if and only if $(\psi^{-1}\varphi, W)$ solve
the equations with seed data $(\Sigma^n; \psi^{-q_n-2}\lambda, \psi^{-q_n}N, \psi^{-2}\sigma, \tau, \psi^{-\frac32 q_n+1}r, \psi^{-q_n}j)$.}, which is not the case for the traditional conformal constraint equations (obtained from \eqref{origLich}-\eqref{origVect} by setting $N=\frac{1}{2}$.)
In particular, as long as the conformal factor relating a metric to $\lambda$ is bounded above and below, we can choose to work with that conformally related metric instead. 

The utility of the conformal method depends upon the extent to which one can determine and classify those sets of seed data for which the conformal constraint equations admit a unique solution, and those sets for which this is not the case. For constant mean curvature (CMC) seed data on closed manifolds (so long as the nongravitational fields are either Maxwell, perfect fluids, massless scalar fields, or vacuum, with nonnegative cosmological constant), this can be done completely (see \cite{Isenberg95}), primarily because the CMC condition $d \tau=0$ decouples the conformal constraint equations \eqref{origLich}-\eqref{origVect}, and because the Yamabe classification of Riemannian metrics on closed manifolds is well understood (see \cite{LP87}). For seed data sets on closed manifolds which are nearly CMC in an appropriate sense, the determination of which sets lead to (unique) solutions and which do not is essentially complete as well (see \cite{IM96,CBIY00,ACI08, HMM16}). 

In sharp contrast, for seed data without any restriction on the mean curvature, much less is known. Moreover, it has been shown that for some seed data sets on closed manifolds there are multiple solutions, and for some families of seed data sets one can pass from sets with no solutions to sets with unique solutions and on to sets with many solutions (see \cite{Maxwell11, Nguyen15}). The general picture for seed data which is neither CMC nor near-CMC is very unclear.

What about seed data which is asymptotically Euclidean (AE)? While the analytical simplifications of the conformal constraint equations which result from working with either CMC or near-CMC data sets occur regardless of whether the data is specified on a closed manifold or is asymptotically Euclidean (or is asymptotically hyperbolic), the Yamabe classification of AE metrics is  more complicated (and less intuitive) than the Yamabe classification of metrics on closed manifolds. For example, it has been shown (see \cite{DM15}) that an asymptotically Euclidean metric can be conformally deformed to an AE metric with zero scalar curvature if and only if the metric is Yamabe positive (as defined via the Yamabe invariant \eqref{YamInv} below). As well, it has been shown that an AE metric is Yamabe null if and only if it can be conformally deformed to a metric with scalar curvature $\mathcal{R}$ for \emph{every}  function $\mathcal{R} \le 0$ \emph{except} $\mathcal{R} \equiv 0$. As a consequence of these features (and others) of the AE Yamabe classes, the analysis of the existence and nonexistence of solutions of the conformal constraint equations is generally more complicated for asymptotically Euclidean seed data than it is for seed data on closed manifolds. 

The very recent advances in our understanding of the Yamabe classes of metrics for asymptotically Euclidean metrics  (see \cite{DM15}), besides indicating some of the difficulties of the analysis of the conformal constraint equations for AE data, also provide information which is useful in handling these difficulties. In this paper, after a brief review (in Section \ref{AE}) of asymptotically Euclidean geometries and their analytic features (properties  of Fredholm operators on AE geometries, the various AE maximum principles and sub and supersolution theorems, and the AE Yamabe classes), we discuss a number of new results concerning solutions of the conformal constraint equations for various classes of AE seed data.  

The key to many of the results we present here is the Curvature Criterion Theorem which we discuss and prove in Section \ref{CurvCrit}. This result (Thm \ref{LichIff}, below) states (roughly) that the (stand alone) \emph{Lichnerowicz equation}---which we obtain by replacing the coefficient of the $\varphi^{-q_n -1}$ term by an arbitrary (specified) nonnegative function $f^2$---admits a positive solution $\varphi$ with appropriate  falloff properties if and only if the metric $\lambda$ admits a conformal transformation $\lambda \rightarrow \psi^{q_n -2} \lambda$ such that the scalar curvature corresponding to $\psi^{q_n -2} \lambda$ is equal to $-\kappa_n \tau^2$. From this result, we readily infer an Admissibility Corollary (Corollary \ref{Admiss} in Section \ref{CurvCrit}) for the conformal constraint equations; that is, we infer conditions on the seed metric $\lambda$ and the mean curvature function $\tau$ which must hold if the conformal constraint equations \eqref{origLich}-\eqref{origVect} are to admit a solution for a given set of seed data $(\Sigma^n; \lambda, \sigma, \tau, r, j)$. We note that both the Curvature Criterion Theorem and the Admissibility Corollary, which we present and prove in this paper for asymptotically Euclidean data sets, have been inspired by earlier work of Maxwell \cite{Maxwell05}, in which he proves analogous results for initial data (with Yamabe negative metrics) on closed manifolds. 

Many of the new insights obtained in \cite{DM15} concerning the Yamabe classification of asymptotically Euclidean metrics involve the prescribed scalar curvature problem for conformal deformation of AE metrics. Since this problem plays a major role in the Admissibility Corollary, we are immediately led to simple results regarding the existence and nonexistence of solutions to  \eqref{origLich}-\eqref{origVect} for various classes of seed data. These are stated in Sections \ref{CurvCrit} and \ref{NexBlow}. Included in these sections are comments regarding the existence of seed data sets of each of these classes. 

Knowing that there are seed data sets $\hat{\mathcal{S}}:= (\Sigma^n; \hat \lambda, \hat \sigma, \hat \tau, \hat r, \hat j)$ for which no solutions to the conformal constraint equations exist, it is useful to consider sequences of seed data sets $\mathcal{S}_\ell$ which approach  $\hat{\mathcal{S}}$. In Section \ref{NexBlow}, we prove a pair of results showing that if a solution to Eqns.\ \eqref{origLich}-\eqref{origVect} exists for each element of the seed data sequence $\mathcal{S}_\ell$, then these solutions must blowup---in the sense that $\sup \varphi_\ell \rightarrow \infty$---as $\mathcal{S}_\ell$ approaches $\hat{\mathcal{S}}$. 

As argued in Section \ref{NexBlow}, for asymptotically Euclidean seed data sets which are CMC (and consequently maximal, with $\tau=0$), if the Yamabe class of the metric $\lambda$ is null or negative, then the conformal constraint equations admit no solutions. To date, all existence theorems for 
solutions of Eqns.\ \eqref{origLich}-\eqref{origVect} for AE seed data (see \cite{Maxwell05b, CBIY00, DIMM14}) have also required that the Yamabe class of the seed metric be positive.  In Section \ref{zeros}, we prove an existence result which allows for AE seed data with nonpositive Yamabe class. As well, this result allows for $\tau^2$ to have zeroes. This is significant because the known existence results for  nonpositive Yamabe class seed data on compact manifolds all require $\tau^2>0$.

\section{Asymptotically Euclidean Initial Data Sets}
 \label{AE}

In  working with asymptotically Euclidean geometries and initial data sets, we use the definitions and conventions of \cite{Bartnik86}. Specifically, we first define a $C^{\infty}$, $n$-dimensional Riemannian manifold $(\Sigma^n; e)$ to be \emph{Euclidean at infinity} if there exists a compact subset $V \subset \Sigma^n$ such that $\Sigma^n \setminus V$ is the disjoint union of a finite number of open sets $U_k$, and each $(U_k; e|_{U_k})$ (called an \emph{end}) is isometric to the exterior of a ball in Euclidean space. Associated to each end is a natural, smooth radial function. We smoothly interpolate these to define a function $\rho\geq1$, which is precisely the coordinate radial function on each end. We then define the weighted Sobolev spaces $W^{k,p}_\delta (\Sigma^n,e)$ (with $1\le p \le \infty, s\in \mathbb{N}^+$, and $\delta \in \mathbb{R}$) of tensor fields on $(\Sigma^n; e)$  to be the closure of the set of $C^{\infty}_0$ tensors with the respect to the norm
\begin{equation}
\label{spdeltanorm}
\lVert T \rVert_{W^{k,p}_\delta} :=\Bigg\{\sum_{0\le m \le k} \int_{\Sigma^n} |\nabla^mT|^p \rho^{-n + p(m-\delta)} \nu \Bigg\} ^{\frac{1}{p}},
\end{equation}
where $\nabla^m$ denotes the $m$'th-order covariant derivatives compatible with the metric $e$ (arbitrarily smoothly extended into the interior region $K$ as a Riemannian metric), $|\ |$ is the corresponding tensor norm and $\nu$ is the corresponding volume element. Note that by our choice of norm, $\delta$ directly encodes the falloff of $T$ on the ends. Based on these two definitions, we define a Riemannian manifold $(\Sigma^n; \gamma)$ to be \emph{$W^{k,p}_\delta$-asymptotically Euclidean} if there exists a Riemannian metric $e$ on $\Sigma^n$ such that $(\Sigma^n; e)$ is Euclidean at infinity, and such that the tensor $\gamma-e$ is contained in $W^{k,p}_\delta (\Sigma^n,e)$.   For simplicity, we shall often call a Riemannian manifold asymptotically Euclidean\footnote{In working with the conformal constraint equations in CTS-H form for a specified set of seed data $(\Sigma^n, \lambda, N, \sigma, \tau, r, j)$ the lapse $N$ contains information regarding the geometry, along with $\lambda.$ Hence, the asymptotically Euclidean condition for seed data must include the  requirement that $N-1 \in W^{k,p}_\delta (\Sigma^n,e)$. More generally, the regularity and asymptotic properties of $N$ should match that of $\lambda$.}  without referring to the specific values of $k, p,$ and $\delta$. We note, however, that unless otherwise specified, we assume that $k > \frac{n}{p}$ so that the metric $\gamma$ is continuous and we assume that $\delta < 0$ so that $\gamma$ approaches $e$ asymptotically on each end. We also note that, if $(\Sigma^n; \gamma)$ is $W^{k,p}_\delta$-AE, we may replace the norm \eqref{spdeltanorm} by an equivalent norm defined using the metric $\gamma$ rather than $e$. For later use, we also define AE H\"older norms by
\begin{equation}
\|T\|_{C^{k,\alpha}_{\delta}} := \|T\rho^{-\delta}\|_{C^{k,\alpha}}.
\end{equation}

An initial data set $(\Sigma^n; \gamma, K)$ for the Einstein vacuum equations is defined to be asymptotically Euclidean if the Riemannian manifold $(\Sigma^n; \gamma)$ is $W^{k,p}_\delta$-AE, and if in addition the tensor field $K$ is an element of $W^{k-1,p}_{\delta-1}$, where $k \geq 1$ and $\delta \in (2-n,0)$. If there are non-gravitational fields present, these are generally expected to be elements of $W^{k'-1,p}_{\delta '}$, for some $k'$ and $\delta '$ related to $k$ and $\delta$, although discontinuities on co-dimension one submanifolds are often allowed for fluid fields.  We note that a seed data set is defined to be AE if it satisfies conditions analogous to these. 
%We discuss this further below.

Our analysis of the conformal constraint equations and their solutions for AE seed initial data sets are carried out using $W^{k,p}_\delta$ tensor fields; it is therefore useful to recall some of the properties of these spaces. We state these properties in the form of the following two lemmas, both of which are proven in \cite{CBC81}:
%In order to help us find solutions to the coupled equations, we will work with the subcritical system
%\begin{equation}\label{lich}
%\alpha_n \Delta \varphi = R\varphi +\kappa_n \tau^2 \varphi^{N-1} - |\sigma +LW|^2 \varphi^{-N-1},
%\end{equation}
%\begin{equation}\label{vect}
%\di LW = \kappa_n  \varphi^{N-\epsilon} \tau^{\epsilon} d\tau.
%\end{equation} Notice that this is different than in \cite{DGH11, GS12}, in that we also include $\tau^{\epsilon}$. These equations are not subcritical in any technical sense, like for the Yamabe problem. The $\epsilon$s help us be able to find constant supersolutions to the Lichnerowicz equation. As shown in \cite{Maxwell09}, a supersolution is all that is needed in order to find solutions to the constraint equations, but it has proven to be the difficult part of the construction. The trouble we go through to show convergence as $\epsilon\to 0$ is because we need this supersolution. In particular, the $\epsilon$ on $\varphi$ allows us to control the max of $\varphi$, while the $\tau^\epsilon$ allows us to get slightly stronger falloff for $LW$.
%We next list some known results about asymptotically flat function spaces and existence of solutions to the conformal constraint equations.

\begin{lem}
\label{embeddings}  \emph{(Sobolev Embeddings)\textbf{.}}

\begin{enumerate}

\item For $k'>k, \delta' < \delta$ and $1\leq p < \infty$, the inclusion $W^{k',p}_{\delta'} \subset W^{k,p}_{\delta}$ is compact.

\item For $k'<k-n/p$, the inclusion $W^{k,p}_{\delta} \subset C^{k'}_{\delta}$ is compact.
\end{enumerate}
\end{lem}

\begin{lem} [Sobolev Multiplication]
\label{multiplication}
If $m\le\min(j,k)$, $p\le q$, $\epsilon>0$, and $m<j+k-n/q$, then
multiplication is a continuous bilinear map from
$W^{j,q}_{\delta_1}\times W^{k,p}_{\delta_2}$ to
$W^{m,p}_{\delta_1+\delta_2+\epsilon}$ for any $\epsilon>0$. In particular,
if $k>n/p$ and $\delta<0$, then $W^{k,p}_{\delta}$ forms an algebra.
\end{lem}

We next consider elliptic differential operators and solutions of elliptic differential equations on AE manifolds. For general elliptic differential operators, we refer the reader to \cite{CBC81}, from which our standard elliptic type estimates come.
Focussing on PDEs of the form \eqref{origLich}-\eqref{origVect}, we have for the latter of these two equations\footnote{Note that here and below, we use $\lambda$ to denote a Riemannian metric, since the analysis in solving the conformal constraint equations presumes that a set of seed data has been chosen, including a metric $\lambda$.}  (as proven in Section IV of \cite{CBIY00})
\begin{lem} [Solvability of Vector-Laplacian Equations]
\label{VectLap}
Let $(\Sigma^n; \lambda)$ be $W^{k, p}_{\rho}$-AE with $k > n/p$, $k\geq 2$ and $\rho <0$, let the lapse $N$ satisfy the condition that $N-1 \in W^{k,p}_\rho$, and let the vector field $\omega \in W^{k-2, p}_{\delta -2}$ with $\delta \in (2-n,0)$.  Then the ``vector-Laplacian" equation
\begin{equation}
\di_{\lambda}\Big(\frac{1}{2N} L_{\lambda}W\Big) = \omega
\end{equation}
has a unique solution  $W \in W^{k,p}_{\delta}$, satisfying
\begin{equation}
  \|W\|_{W^{k,p}_{\delta}} \leq c\|\omega\|_{W^{k-2,p}_{\delta-2}}
\end{equation} for some constant $c$ which depends only on the metric $\lambda$ and the lapse $N$.
\end{lem}
We note that if we choose $\tau \in W^{k-1,p}_{\delta-1}$, $\varphi \in W^{k,p}_{loc}\cap L^\infty$, and $j \in W^{k-2, p}_{\delta -2}$, and if we set $\omega=  \kappa_n  \varphi^{q_n } d\tau + j$ as in \eqref{origVect}, then the hypothesis of Lemma \ref{VectLap} is satisfied.

To work with PDEs of the Lichnerowicz form \eqref{origLich}, we wish to establish a sub and supersolution theorem for equations of that form. To do that, we first recall some results from \cite{Maxwell05b}:

\begin{lem} [$-\Delta +V$ as a Fredholm Operator]
\label{AELinearExistence}
Let  $(\Sigma^n; \lambda)$ be $W^{k,p}_\rho$-AE with $k\geq 2$, $k>n/p$ and $\rho<0$, let $\Delta$ be the corresponding (negative eigenvalue) Laplacian operator, and suppose that the function $V$ is contained in the space $W^{k-2,p}_{\rho-2}$. If $\delta\in (2-n,0)$, then the operator
$\mathcal{P}:= -\Delta +V: W^{k,p}_\delta \to W^{k-2,p}_{\delta-2}$ is Fredholm with index 0. Moreover, if $V\geq 0$
then $\mathcal{P}$ is an isomorphism, in which case the standard elliptic estimate holds in these spaces.
\end{lem}
This lemma corresponds to Proposition 1 in \cite{Maxwell05b}. Note that the hypothesis for this lemma places no restrictions on $\rho$ other than that $\rho<0$, which is needed  so that the metric decays to the Euclidean metric. The following is a slight strengthening of a maximum principle in \cite{Maxwell05b}.

\begin{lem}[A Maximum Principle for AE Manifolds]
\label{MaxPrinciple1}
Suppose that $(\Sigma^n; \lambda)$ and $V$ satisfy the hypothesis of  Lemma \ref{AELinearExistence}, and 
suppose that $V\geq 0$. If $u \in W^{2,p}_{loc}$, if for any choice of positive $\eta$ one has $u\geq -\eta$ outside of a sufficiently large ball, and if $u$ satisfies the differential inequality 
\begin{equation}
\label{diffineq}
-\alpha_n\Delta u + V u \geq 0,
\end{equation}
then $u\geq 0$.
\end{lem}
\begin{proof}
Let 
\begin{equation}
\label{defv}
v = (u + \epsilon)^- := \min\{0,u+\epsilon\}
\end{equation} 
for some $\epsilon>0$. It follows from this definition and from assumptions made above that $v$ is compactly supported, and that $v \leq0$. As well, it follows from Sobolev embedding (see Lemma \ref{embeddings}) that  $v\in W^{1,2}$. Using integration by parts (for $v$ with compact support), using the differential inequality \eqref{diffineq}, and using the fact that wherever $v$ is nonzero (and therefore negative) it must be the case that $u$ is negative (see \eqref{defv}), we have 
\begin{equation}
\int_{\Sigma^n} \alpha_n|\nabla v|^2 = \int_{\Sigma^n} -\alpha_n v \Delta u \leq \int_{\Sigma^n} -Vu v \leq  0.
\end{equation} 
This tells us that $v$ must be constant on $\Sigma^n$. Since we know that there are places on $\Sigma^n$ where $u\geq -\eta$ for any positive $\eta$, it follows from \eqref{defv} that $v$ must be identically zero. Hence $u\geq -\epsilon$ for any positive $\epsilon$. Letting $\epsilon \to 0$, we determine that $u \geq 0$.
\end{proof}

For use in proving the Curvature Criterion Theorem below, we note one further maximum principle, which is Lemma 4 from \cite{Maxwell05b}:

\begin{lem}[Another Maximum Principle for AE Geometries]
 \label{MaxwellMax}
  Suppose that  $(\Sigma^n; \lambda)$ is $W^{k,p}_\delta$-AE with $k\geq 2$, $k>n/p$, and $\delta<0$, and suppose that $V \in W^{k-2,p}_{\delta-2}$, and suppose that $u\in W^{k,p}_{loc}$ is nonnegative and satisfies $-\Delta u + Vu \geq 0$ on $\Sigma^n$. If $u(x) = 0$ at some point $x\in \Sigma^n$, then $u$ vanishes identically.
\end{lem}

To state the sub and supersolution theorem in a form which can be applied to the Lichnerowicz equation, we consider the nonlinear PDE
\begin{equation}
\label{eq:Semilinear}
-\alpha_n\Delta \varphi = F(x,\varphi)
\end{equation}
for a function $F:\Sigma^n \times \R\to \R$ which takes the form 
$F(x,y) = \sum_{i=1}^j a_i(x) y^{b_i}$ for specified functions $a_i:\Sigma^n \to \R$
and for constants
$b_i$, where we use the convention that $y^{b_i} \equiv 1$ if $b_i =0$.
We assume here that $a_i(x) \in W^{k-2,p}_{\delta-2}$ for some $k>n/p$, $k\geq 2$ and $\delta\in (2-n,0)$. We note that, depending on the value(s) of $b_i$, the quantity
$y^{b_i}$ is smooth on $(0,\infty)$, $[0,\infty)$, or $(-\infty,\infty)$. 
We define the function
$F$ to be  ``Lichnerowicz-type" if it satisfies these properties, and we define the   largest interval for which
all the $y^{b_i}$ are smooth to be  $F$'s ``interval of regularity" $I$. 

Recalling that a pair of functions $\varphi_-$ and $\varphi_+$ are called sub and supersolutions of an equation of the form \eqref{eq:Semilinear} if $-\alpha_n\Delta \varphi_- \leq F(x,\varphi_-)$, if $-\alpha_n\Delta \varphi_+ \geq F(x,\varphi_+)$, and if $\varphi_-(x) \leq \varphi_+(x)$, we have the following existence theorem: 

\begin{thm}[Sub and Supersolution Theorem for AE Manifolds]
\label{SubSup}
Let $(\Sigma^n; \lambda)$ be $W^{k,p}_\rho$-AE with $k>n/p$, $k\geq 2$ and
$\rho<0$. Suppose that $F(x,y)$ is Lichnerowicz-type (as defined above) for some $\delta \in (2-n,0)$.
Suppose that there are sub and supersolutions $\varphi_\pm\in W^{2,p}_{loc} \cap L^\infty$  for which $\inf \varphi_- \in I$. Suppose there is a constant $\xi>0$ such that,
sufficiently far
out on each end, $\varphi_- \leq \xi \leq \varphi_+$.
Then Equation \eqref{eq:Semilinear} admits a solution $\varphi$ such that
$\varphi_- \leq \varphi \leq \varphi_+$ and such that $\varphi- \xi \in W^{k,p}_{\delta}$.
\end{thm}

\begin{proof}
We prove this lemma using the strategy established for the analogous result in  \cite{CBIY00}. That is, we construct a solution by induction, starting from $\varphi_-$. Let $s$ be a positive
function on $\Sigma^n$ such that $s \in L^p_{\delta}$ and
\begin{equation}\label{eq:kDef}
s(x) + \sup_{y \in \textrm{Range}(\varphi_\pm)} F_y(x,y) \geq 0.
\end{equation} Such a function $s$
exists as a consequence of our assumptions on $\varphi_\pm$ and on $F$. 

We now define $v_1\in W^{2,p}_{\delta}$ to be the unique solution to
\begin{equation}
\label{defv1}
-\alpha_n\Delta v_1 + sv_1 = F(x,\varphi_-) + s(\varphi_- -\xi)
\end{equation} and set $\varphi_1 := v_1 + \xi$. The solution $v_1$ exists as a consequence of Lemma
\ref{AELinearExistence}. 

Using Eqn.\ \eqref{defv1} satisfied by $v_1$ and the subsolution differential inequality satisfied by $\varphi_-$,  we calculate
\begin{equation}
-\alpha_n\Delta(\varphi_1- \varphi_-) +s(\varphi_1-\varphi_-) \geq 0.
\end{equation} 
It then follows from the maximum principle \ref{MaxPrinciple1} that $\varphi_1 \geq \varphi_-$.
Similarly, we calculate 
\begin{align}
-\alpha_n\Delta(\varphi_+ - \varphi_1) + s(\varphi_+-\varphi_1) &\geq F(x,\varphi_+) - F(x,\varphi_-) + s(\varphi_+-\varphi_-)\\
&= (\varphi_+-\varphi_-) \left(s + \int_0^1 F_y(x,\varphi_-+t(\varphi_+-\varphi_-)) dt\right)\\
&\geq 0,
\end{align} 
where the last line holds as a consequence of the  assumption \eqref{eq:kDef} which we have used in defining $s$. 
Again, it follows from  the maximum principle \ref{MaxPrinciple1} that $\varphi_1\leq \varphi_+$.

With $\varphi_1$ to initialize it, we now define the sequence  $\varphi_i = v_i + \xi$, with  $v_i \in W^{2,p}_\delta$ solving
\begin{equation}
-\alpha_n\Delta v_i + s v_i = F(x,\varphi_{i-1}) - s v_{i-1}.
\end{equation} 
Again using the maximum principle, we can show that $\varphi_i$ is an
increasing sequence; i.e.,
\begin{equation}
\varphi_- \leq \varphi_1 \leq \varphi_2 \leq \cdots \leq \varphi_{i-1} \leq \varphi_i \leq \cdots \leq \varphi_+.
\end{equation} Since the functions $\varphi_i$ constitute a bounded increasing sequence, the $\varphi_i$ converge
to some function $\varphi$, with $\varphi_- \leq \varphi \leq \varphi_+$. We claim that this limit function $\varphi$ is a solution of Eqn.\ \eqref{eq:Semilinear}.

To show this, we start by noting that  the elliptic estimates indicated by Lemma \ref{AELinearExistence} give us
\begin{equation}
\|v_{i+1}\|_{W^{2,p}_\delta} \leq C \|F(x,\varphi_i) - s v_i\|_{L^p_\delta}.
\end{equation} 
Combining our assumptions on $s$ and $F$ with the bounds noted above for $v_i$ and $\varphi_i$, we verify that 
the right hand side of this inequality is uniformly bounded.  It thus follows that $v_i$ is uniformly
bounded in $W^{2,p}_\delta$.

The compact embedding of $W^{2,p}_\delta$ into 
$C^{0,\alpha}_{\delta'}$ which holds for any $\delta'> \delta$ and for some $\alpha>0$
(see Lemma \ref{embeddings})
implies that $\varphi_i \to \varphi$ in $C^{0,\alpha}_{\delta'}$, and that 
$\varphi-\xi \in W^{2,p}_\delta$.
This convergence implies that $F(x,\varphi_{i-1}) - kv_{i-1}$ converges in $L^{p}_\delta$;
thus, since $-\alpha_n\Delta +s$ is an isomorphism, $\varphi_i$ must converge to $\varphi$ in 
$W^{2,p}_\delta$. Consequently we obtain $-\alpha_n\Delta \varphi = F(x,\varphi)$, as 
desired. Additional regularity can be achieved by the usual bootstrap arguments.
\end{proof}

Lemma \ref{SubSup} is very useful for proving that  Eqn.\ \eqref{origLich} admits solutions \emph{if} it is decoupled from Eqn.\ \eqref{origVect}. For maximal seed data sets (those with $\tau=0$), this is the case.\footnote{With $j$ nonzero, the system is coupled, but the coupling is sequential: One can solve  \eqref{origVect} for $W$ independent of $\varphi$, then substitute the solution $W$ into \eqref{origLich} and finally solve the resulting equation for $\varphi$.} For the fully coupled system \eqref{origLich}-\eqref{origVect}, Lemma \ref{SubSup} cannot be directly applied. However,  a modified version of it is very useful. We consider \emph{Lichnerowicz-coupled-type} systems of the form
\begin{align}
\label{Eqn1}
-\alpha_n\Delta \varphi = H(x,W,\varphi)\\
\label{Eqn2}
\di_{\lambda} L_{\lambda}W = G(x, \varphi)
\end{align}
where $H(x,W,\varphi)$ satisfies the properties stated above for Lichnerowicz-type functions $F(x, y)$, but with the coefficients $a_j(x)$ in the expansion $F(x,y) = \sum_{i=1}^j a_i(x) y^{b_i}$ allowed to depend on $W$ and its first derivatives, and where $G(x,\varphi)$ is a polynomial in $\varphi$, with $x$-dependent coefficients. We define a \emph{global subsolution and global supersolution} for the system \eqref{Eqn1}-\eqref{Eqn2} to be pair of functions $\Phi_-$ and $\Phi_+$ such that for all $\varphi$ satisfying\footnote{We also require that $ l \leq \Phi_-(x) \leq \Phi_+ (x) \leq m$ for  $[l,m]\subset I$, where $I$ is the interval of regularity for $F$.}
 $\Phi_-\leq \varphi \leq \Phi_+$, one has $-\alpha_n\Delta \Phi_+ \geq H(x,W, \Phi_+)$ and $-\alpha_n\Delta \Phi_- \leq H(x,W, \Phi_-)$  where $W$ satisfies $\di_{\lambda} L_{\lambda}W = G(x, \varphi)$. Relying on this definition and the Schauder Fixed Point Theorem, one can prove as in \cite{Maxwell09} a result  essentially the same as Lemma \ref{SubSup}, showing that if global sub and supersolutions exist, then the system \eqref{Eqn1}-\eqref{Eqn2} admits a solution. 
 
\begin{thm}[Global Sub and Supersolution Theorem]
\label{GlobalSubAndSup}
If such a global sub and supersolution pair exist, then the system \eqref{Eqn1}-\eqref{Eqn2} has a solution $(\varphi, W)$, with $\varphi-\xi$ and $W$ in $W^{k,p}_{\delta}$.
\end{thm}

We recall (see the Introduction) that as a consequence of the conformal covariance of the conformal constraint equations (in CTS-H form), in verifying the existence of solutions to \eqref{origLich}-\eqref{origVect} for a given set of seed data, one may work instead with conformally-related seed data. As seen below, it often makes it easier to find global sub and super solutions if a strategic conformal transformation is implemented at the start of the analysis. 

We close this section with a discussion of the Yamabe classes for asymptotically Euclidean geometries. As with Riemannian geometries on closed manifolds, the Yamabe class for an AE geometry is determined by the sign of the Yamabe invariant. For a specified $W^{k,p}_{\delta}$ AE geometry $(\Sigma^n; \lambda, N)$ (with $k>\frac{n}{p}, k\ge 2$, and $\delta<0$), we define the Yamabe invariant $Y(\Sigma^n, \lambda)$ as follows:
\begin{equation}
\label{YamInv}
Y(\Sigma^n, \lambda):= \inf_{u\in C^\infty_c(\Sigma^n)} \frac{\int_{\Sigma^n} \alpha_n|\nabla u|^2 +R_\lambda u^2}{\|u\|^2_{L^q_n }}.    
\end{equation}
Here $R_\lambda$ denotes the scalar curvature of the metric $\lambda$, and $C^\infty_c(\Sigma^n)$ denotes the set of smooth functions on the AE manifold $\Sigma^n$ with compact support. 
One verifies that the Yamabe invariant $Y(\Sigma^n, \lambda)$ is invariant under conformal transformations of the metric, and one also readily verifies that the three Yamabe classes $\mathcal Y^+(\Sigma^n), \mathcal Y^0(\Sigma^n)$, and $\mathcal Y^-(\Sigma^n)$ partition the set of all AE geometries (and their conformal equivalence classes) on a given manifold $\Sigma^n$. 

For closed $\Sigma^n$, one has the familiar and intuitive---and very useful---result that a metric $\lambda$ can be conformally transformed to a metric $\hat \lambda$ with $R_{\hat \lambda}>0$ if and only if $\lambda \in \mathcal Y^+$; similarly one can conformally transform to $R_{\hat \lambda}=0$ iff $\lambda \in \mathcal Y^0$, and to 
$R_{\hat \lambda}<0$ iff $\lambda \in \mathcal Y^-$. As noted in the Introduction, for asymptotically Euclidean metrics, the results for conformal transformations to metrics with scalar curvature of a prescribed sign are significantly more complicated and much less intuitive. As proven in \cite{Maxwell05b} and in \cite{DM15}, one has the following for AE metrics:
\begin{lem}[AE Yamabe classes and their properties] \quad
\label{Yamabe}

\begin{itemize}
\item $\lambda \in \mathcal Y ^+$ iff it can be conformally deformed to a metric with scalar curvature $\mathcal R$ for every function $\mathcal{R} \le 0.$

\item $\lambda \in \mathcal Y ^0$ iff it can be conformally deformed to a metric with scalar curvature $\mathcal R$ for every  function $\mathcal{R} \le 0$  except $\mathcal{R} \equiv 0$.

\item $\lambda \in \mathcal Y ^-$ iff there exists some  $\mathcal{R} \le 0$, with $\mathcal{R} \not \equiv 0$, such that $\lambda$ cannot be conformally deformed to a metric with scalar curvature $\mathcal R$.

\item If the scalar curvature of a metric $\lambda$ (or of a metric conformally related to $\lambda$) is nonnegative, then $\lambda \in \mathcal Y^+$. 

\item An AE geometry $(\Sigma^n; \lambda)$ is contained in $\mathcal Y ^+$ or $\mathcal Y ^0$ or $\mathcal Y ^-$ iff $(\Sigma^n; \lambda)$ admits a conformal compactification to a geometry of the same Yamabe class. 
\end{itemize} 
\end{lem}

\section{Curvature Criterion for Asymptotically Euclidean Solutions of the Lichnerowicz Equation}
\label{CurvCrit}

For asymptotically Euclidean initial data sets which have constant mean curvature (and consequently are maximal, with $trK = \tau = 0$), the system \eqref{origLich}-\eqref{origVect} decouples, and the conformal method admits a solution for a given set of seed data if and only if Eqn.\ \eqref{origLich} admits a solution. Presuming that the coefficient of the $\varphi^{q_n -1}$ term vanishes (which is true for maximal data with vanishing cosmological constant and with no scalar fields present), one readily verifies that \eqref{origLich} admits a solution if and only if the seed data metric $\lambda$ is in the positive Yamabe class.\footnote{Perhaps the easiest way to see this is to note from the constraint equation \eqref{HamConstr} that for maximal seed date the conformal method reduces to a prescribed scalar curvature problem, with that prescribed scalar curvature being nonnegative.} 

For non-CMC data, one must work with the coupled system. However, even in the non-CMC case it is still useful (as discussed below) to consider the solvability of \eqref{origLich} in its decoupled form (which is often labeled the \emph{Lichnerowicz equation}):
\begin{equation}
\label{LichneroEqn}
\alpha_n \Delta_{\lambda} \varphi = R_{\lambda} \varphi +\kappa_n \tau^2 \varphi^{q_n -1} -f^2 \varphi^{-q_n -1} - r \varphi ^ {-\frac{q_n }{2}}.
\end{equation}
Relying primarily on the sub and supersolution theorem, as stated above in Lemma \ref{SubSup}, we prove the following:
\begin{thm} [Curvature Criterion for AE Solutions to the Lichnerowicz Equation]
\label{LichIff}
Suppose that $(\Sigma^n; \lambda)$ is $W^{k,p}_\delta$-AE 
with $k> n/p$, $k\geq 2$ and $\delta<0$, and suppose that  $r\geq 0$, that $r, f^2$ and $\tau^2$ are in $W^{k-2,p}_{\delta-2}$. The Lichnerowicz equation \eqref{LichneroEqn} has a positive solution
$\varphi$ with $\varphi-1\in W^{k,p}_\delta$ if and only if there exists a conformal factor $\psi>0$ with
$\psi-1\in W^{k,p}_\delta$ such that $R_{\psi^{q_n -2}\lambda} = -\kappa_n  \tau^2$.

The same result holds if we replace the conditions and conclusions stated here for $\varphi-1$ and for $\psi-1$ by conditions and conclusions imposed on  $\varphi-c_1$ and on $\psi-c_2$ for any positive constants $c_1$ and $c_2$.
%Indeed, if either of the equations have a positive solution $\varphi$ such that
%$\varphi-c\in W^{k,p}_\delta$ for some $c>0$, then both of the equations have solutions with any
%desired constant.
\end{thm}

The analogous  theorem for data on closed manifolds, which holds for metrics  the Yamabe negative 
class, was originally proven by Maxwell in \cite{Maxwell05}, using a very similar
proof. 

\begin{proof}
$(\Rightarrow)\,$ We presume that a solution $\varphi$ to the Lichnerowicz equation \eqref{LichneroEqn} exists, with regularity and asymptotics as stipulated in the hypothesis.  It follows from the formula for the transformation of the scalar curvature induced by a conformal transformation of the metric that to prove the first part of this theorem, it is sufficient that we show
that there exists a solution $\psi$ (with appropriate asymptotic and regularity properties) to
the equation
\begin{equation}
\label{eq:conformalFactor}
-\alpha_n\Delta _\lambda \psi + R_\lambda \psi + \kappa_n \tau^2 \psi^{q_n -1}= 0.
\end{equation}

To apply Lemma \ref{SubSup} to Eqn.\ \eqref{eq:conformalFactor}, we first note that it follows from the hypothesis of Theorem \ref{LichIff} that $R_{\lambda} \in W^{k-2,p}_{\delta-2}$; thus we verify that this hypothesis ensures that the coefficients of the terms in \eqref{eq:conformalFactor} satisfy  the regularity requirements for Lemma \ref{SubSup}. For the supersolution $\psi_+$ for \eqref{eq:conformalFactor}, we choose
$\varphi$, the solution of the Lichnerowicz equation, with its prescribed regularity. For the
subsolution, we take $\psi_- \equiv 0$, which is of course sufficiently regular. Noting that the powers of $\psi$ appearing in the function $F(x,\psi) = R_\lambda \psi + \kappa_n  \tau^2 \psi^{q_n -1}$ are all positive, we verify that indeed $\psi_-$ lies within the interval of regularity for $F$. It thus follows from Lemma \ref{SubSup} that \eqref{eq:conformalFactor} admits a smooth solution $\psi$ bounded between $\psi_-$ and $\psi_+$. One must still verify that this solution is bounded away from zero. This follows immediately from the maximum principle Lemma \ref{MaxwellMax}, along with the requirement that $\psi$ approach $1$ at infinity.

$(\Leftarrow)\,$ We suppose now that there is a conformal factor $\Psi$ with the stipulated regularity and asymptotic behavior for which $R_{\Psi^{q_n -2}\lambda} = -\kappa_n  \tau^2$. We need to show that there must exist a solution of the Lichnerowicz equation \eqref{LichneroEqn}, which for convenience we rewrite in the (nonlinear operator) form
\begin{equation}
\label{LichneroOp}
0=\mathcal{L} (\Phi) := -\alpha_n \Delta_{\lambda} \Phi + R_{\lambda} \Phi +\kappa_n \tau^2 \Phi^{q_n -1} -f^2 \Phi^{-q_n -1} - r \Phi ^ {-\frac{q_n }{2}}.
\end{equation}
The conformal covariance of the conformal constraint equations carries over to the Lichnerowicz equation, as long as we transform $f^2$ as $\hat f^2 = \Theta^{-2q_n} f^2$. Using the function $\Psi$ as our conformal factor, and denoting the conformally transformed quantities by hats, we see that it is sufficient to show that there exists a solution to 
\begin{equation}
\label{LichneroOp2}
0=\hat{\mathcal{L}} (\Phi) := -\alpha_n \Delta_{\hat\lambda} \Phi + \kappa_n \tau^2 (\Phi^{q_n -1} - \Phi)
 -\hat f^2 \Phi^{-q_n -1} - \hat r \Phi ^ {-\frac{q_n }{2}}.
\end{equation} 
We readily verify that if $\hat f^2$ and $\hat r$ both vanish everywhere, then $\Phi=1$ is a solution to this equation. Hence, we may presume that one or the other of these quantities is nonzero somewhere. We also readily verify that $\hat{\mathcal{L}}(1) \le 0$, so $\Phi_-=1$ is a subsolution. To obtain a supersolution, we note that it follows from 
Lemmas \ref{AELinearExistence}-\ref{MaxwellMax} that the linear equation 
\begin{equation}
\label{lineq}
-\alpha_n\Delta_{\hat \lambda} u + \kappa_n  \tau^2 u = \hat f^2 +\hat r
\end{equation} 
admits a solution $u-1\in W^{k,p}_\delta$, with $u\ge c>0$ for some constant $c$. Calculating $\hat{\mathcal{L}}(\beta u)$ for a positive constant $\beta$, we obtain
\begin{align*}
\hat{\mathcal{L}} (\beta u) = &-\alpha_n \Delta_{ \hat \lambda} (\beta u) + \kappa_n \tau^2 ([\beta u]^{q_n -1} - \beta u)-\hat f^2 [\beta u]^{-q_n -1} - \hat r [\beta u]^ {-\frac{q_n }{2}}\\
 =& \kappa_n  \tau^2[(\beta u)^{q_n -1}-2\beta u] + \hat f^2[\beta-(\beta u)^{-q_n -1}] + \hat r[\beta-(\beta u)^{-q_n /2}].
 \end{align*}
 Since, as noted above, $u$ is bounded away from zero, we see from this calculation that for sufficiently large $\beta$, $\Phi_+ = \beta u$ is a supersolution for \eqref{LichneroOp2}. It follows that the Lichnerowicz equation admits a solution with the stated regularity and asymptotic behavior. 
 
 To prove this result for alternate asymptotic limits for $\varphi$ and for $\psi$, we argue as follows. Starting with a solution $\varphi$ to the Lichnerowicz equation (i.e., going $\Rightarrow$) we note that for any choices of the asymptotic limits $c_1$ and $c_2$, there is a constant $\hat \beta \ge 1$, for which $\hat \psi_+ = \hat \beta \varphi$ is a supersolution for the conformal transformation equation \eqref{eq:conformalFactor} with the prescribed limits. Using $\hat \psi_- =0$ as above, we can apply  Lemma \ref{SubSup} and thereby obtain a solution $\psi $ of \eqref{eq:conformalFactor}. We argue as above using Lemma \ref{MaxwellMax} to show that this solution is positive everywhere. 
 
Starting instead with the presumed existence of the conformal transformation $\Phi$ with appropriate limits (i.e., going $\Leftarrow$), we again verify that with appropriate scalings of the subsolution and of the supersolution, we prove the existence of a solution of the Lichnerowicz equation with the desired asymptotic limits.   
 \end{proof}

Given a function $\tau^2\in L^p_{\delta-2}$, for which $W^{2,p}_{\delta}$-AE metrics does there exist a conformal transformation $\psi$ such that the scalar curvature of the transformed metric takes the value $-\kappa_n  \tau^2$? As noted in the first point of Lemma \ref{Yamabe}, if $\lambda \in \mathcal Y ^+$, then such a transformation always exists. While this condition is sufficient, it is not necessary. As shown in \cite{DM15}, one has the following:

\begin{lem} [Yamabe Classes and the Curvature Criterion]
\label{Yam&CurvCrit}
Suppose $\tau^2 \in L^p_{\delta-2}$. There exists a conformal transformation function $\psi>0$, $\psi-1 \in W^{2,p}_{\delta}$
such that $R_{\psi^{q_n -2} \lambda} = - \kappa_n  \tau^2$ if and only if one of the following
conditions holds:
\begin{itemize}
\item $\lambda \in \mathcal Y ^+$.

\item $\lambda \in \mathcal Y^0$ and $\tau^2 >0$ on a set of positive measure.

\item $\lambda \in \mathcal Y^-$ and $\tau^2 =0$ on a set of sufficiently small measure (in a sense described
in \cite{DM15}). 
\end{itemize} 
\end{lem}

It follows immediately from this lemma that for any choice of AE seed data $(\Sigma^n; \lambda, N, \sigma, \tau, r, j)$, so long as $\lambda \in \mathcal Y ^+$, the Lichnerowicz equation \eqref{LichneroEqn} (with $f^2 = |\sigma+\frac{1}{2N} L_\lambda W|^2$) admits a solution. This result has long been known, and is directly relevant for the conformal method for maximal ($\tau\equiv0)$ seed data. This lemma also tells us that for any choice of AE seed data with $\lambda \in \mathcal Y ^0$ and with $\tau\not\equiv 0$, the Lichnerowicz equation admits a solution. This consequence, which is new, is not directly relevant to the conformal method, since AE seed data with $\tau\not\equiv 0$ must be non CMC, in which case the full coupled system \eqref{origLich}-\eqref{origVect} must be solved.

Although the solvability of the Lichnerowicz equation for a given set of seed data is not the full story, it \emph{does} serve as an admissibility (or necessary) condition for the conformal method to work for that set of data. This is because Theorem \ref{LichIff} provides conditions on the seed data for \eqref{LichneroEqn} to admit a solution which are independent of the function $f^2$. If these  conditions are not met, then regardless of $\sigma$ and regardless of $LW$, the system \eqref{origLich}-\eqref{origVect} cannot admit a solution. If these conditions are met, then there may or may not be a solution to the conformal constraint equations. We summarize this discussion by stating the following corollary:

\begin{cor} [Admissibility Condition for AE Seed Data Sets]
\label{Admiss}
Let $(\Sigma^n; \lambda, N, \sigma, \tau, r, j)$ be a set of asymptotically Euclidean seed data, with regularity conditions as stated in Theorem \ref{LichIff} (except with $f^2$ replaced by  $|\sigma|^2_\lambda$, and with the added conditions that $N -1 \in W^{k,p}_\delta$ and  $j \in W^{k-2, p}_{\delta-2}$ ). 
The existence of a (suitably regular) conformal transformation of the metric $\lambda$ to a metric $\psi^{q_n -2} \lambda$ with scalar curvature $R_{\psi^{q_n -2}\lambda} = -\kappa_n  \tau^2$ is a necessary condition for the conformal constraint equations \eqref{origLich}-\eqref{origVect} to possibly admit a solution. Correspondingly, if none of the conditions listed in Lemma \ref{Yam&CurvCrit} are satisfied by the seed data, then there is no solution to the conformal constraint equations.
\end{cor}

We use this admissibility condition in the next section to show that there are AE seed data sets $(\Sigma^n;\lambda_0, N_0, \sigma_0, \tau_0, r_0, j_0)$ for which the conformal constraint equations do not admit solutions, and proceed to study the behavior of solutions of Eqns.\ \eqref{origLich}-\eqref{origVect} for sequences of seed data sets which approach $(\Sigma^n; \lambda_0, N_0, \sigma_0, \tau_0, r_0, j_0)$. Before doing this, we state a uniqueness result for solutions of the Lichnerowicz equation with AE seed data.

\begin{prop}[Uniqueness of Solutions to the Lichnerowicz Equation]
\label{uniqueness}
Let $(\Sigma^n; \lambda)$ be $W^{k,p}_\delta$-AE  with $k> \frac{n}{p}, k\ge 2$ and $\delta <0$, and let $r\ge 0, f^2$ and $\tau^2$ be functions contained in $W^{k-2,p}_{\delta-2}$. If $\phi_1$ and $\phi_2$ are both (positive) solutions of the Lichnerowicz equation \eqref{LichneroEqn} and if the asymptotic limits of both $\phi_1$ and $\phi_2$ are the same, then $\phi_1=\phi_2$.
\end{prop}

\begin{proof}
The idea of the proof follows that given in Theorem 8.3 of \cite{CBIP06}:  Recalling the conformal transformation formula for scalar curvature, regardless of what the conformal factors $\phi_1$ and $\phi_2$ are, we have 
\begin{align}
\label{conftrans1}
\alpha_n \Delta_\lambda \phi_1 =& R_\lambda \phi_1 -R_{\phi_1^{q_n -2}\lambda} \phi_1^{q_n -1},\\
\label{conftransf2}
\alpha_n \Delta_\lambda \phi_2 =& R_\lambda \phi_2 -R_{\phi_2^{q_n -2}\lambda} \phi_2^{q_n -1},\\
\label{conftransf12}
\alpha_n \Delta_{(\phi_1)^{q_n -2}\lambda} \frac{\phi_2}{\phi_1} =& R_{\phi_1^{q_n -2}\lambda}  \frac{\phi_2}{\phi_1} -R_{\phi_2^{q_n -2}\lambda} (\frac{\phi_2}{\phi_1})^{q_n -1}.
\end{align}
Solving the first of these equations for $R_{\phi_1^{q_n -2}\lambda}$, and using the assumption that $\phi_1$ is a solution of the Lichnerowicz equation, we have 
\begin{align*}
R_{\phi_1^{q_n -2}\lambda} =& \Big(-\alpha_n \Delta_\lambda\phi_1 + R_\lambda \phi_1\Big)\phi_1^{1-q_n }\\
=& \Big( f^2 \phi_1^{-q_n -1} +r \phi_1^{\frac{q_n }{2}} -\kappa_n  \tau^2 \phi_1^{q_n -1}\Big) \phi_1^{1-q_n },
\end{align*}
along with an analogous equation for $R_{\phi_2^{q_n -2}\lambda}$. If we now substitute these formulas for 
$R_{\phi_1^{q_n -2}\lambda}$ and $R_{\phi_2^{q_n -2}\lambda}$ into Eqn. \eqref{conftransf12}, we obtain 
\begin{equation}
\label{phi1phi2}
-\Delta_{\phi_1^{q_n -2} \lambda} (u-1) + \Xi(\phi_1,\phi_2) (u-1) = 0,
\end{equation}
where $u:=\frac{\phi_2}{\phi_1}$, and where $\Xi(\phi_1, \phi_2)\in L^{p}_{\delta-2}$ is a positive expression involving the known functions $\phi_1$, $\phi_2$, $f^2$, $r$, and the metric.  Since $-\Delta_\lambda +\Xi$ is an isomorphism (see Lemma \ref{AELinearExistence}), and
thus injective, we have  that $u -1\equiv 0$, which implies that  $\phi_1 \equiv\phi_2$, so we have uniqueness. We note that the assumed asymptotic value for the solutions $\phi_1$ and $\phi_2$ do not affect the proof, so long as they are the same.
\end{proof}

\section{AE Seed Data for which the Conformal Method Admits No Solutions, and Blow Up Behavior for Nearby Data}
\label{NexBlow}

While the Admissibility Corollary \ref{Admiss} stated above does not stipulate for which AE seed data sets the conformal constraint equations admit a solution, it does stipulate for which such data sets these equations cannot be solved. Combining it with Lemma \ref{Yamabe}, we obtain the following:

\begin{cor} [Nonexistence for Maximal AE Seed Data with Yamabe Nonpositive Metric]
\label{NonCor}
Let $(\Sigma^n; \lambda, N, \sigma, \tau, r, j)$ be a set of asymptotically Euclidean seed data, with regularity conditions as stated in Theorem \ref{LichIff} and Corollary \ref{Admiss}. If the seed data is maximal (i.e., $\tau \equiv 0$) and if $\lambda \in \mathcal{Y}^0$ or $\lambda \in \mathcal{Y}^-$, then the conformal constraint equations \eqref{origLich}-\eqref{origVect} do not admit a solution. Seed data sets satisfying these conditions do exist. 
\end{cor}

\begin{proof} 
The Admissibility Corollary states that a solution to the conformal constraint equations can exist for a given set of seed data only if the metric can be conformally transformed to one with scalar curvature equal to $-\kappa_n  \tau^2$. For maximal seed data, this means that the metric must admit a conformal transformation to a metric with zero scalar curvature. Since Lemma \ref{Yamabe} says that only Yamabe positive metrics are conformally related to zero scalar curvature geometries, the result follows. 

To verify  that there are in fact seed data sets with Yamabe nonpositive metrics, we first note from \cite{DM15} (as is implied by the last point of Lemma \ref{Yamabe}), that for any closed geometry $(\Sigma^n; \lambda)$, there exists a conformal decompactification (i.e., a blow up at some point $p\in \Sigma^n$) which results in an AE geometry $(\tilde \Sigma^n, \tilde \lambda)$ whose Yamabe class is identical to that of $(\Sigma^n, \lambda)$. Since for ``most" closed manifolds all metrics are contained in $\mathcal{Y}^-$, it follows that one readily constructs Yamabe negative AE geometries. Since the map from general symmetric 2-tensors to those which are divergence-free and trace-tree can always be carried through on negative AE geometries (using the ``York decomposition"), it  follows that AE seed data sets (maximal or not)  with negative Yamabe metrics are readily obtained. AE seed data sets with Yamabe zero metrics are similarly readily obtained; we note the results of Friedrich \cite{Friedrich11} as a related approach to obtaining such sets.

\end{proof}

As an immediate consequence of Corollary \ref{NonCor},  one finds that there exist no AE initial data sets 
$(\Sigma^n; \gamma, K, \rho, J)$ which satisfy the Einstein constraint equations \eqref{HamConstr}-\eqref{MomConstr}, which are maximal, and which have either $\gamma \in \mathcal{Y}^0$ or $\gamma \in \mathcal{Y}^-$. As noted in \cite{DM15}, one readily sees that this result directly follows from the Einstein constraint equations \eqref{HamConstr}-\eqref{MomConstr}, along with the statement (see Lemma \ref{Yamabe}) that if an AE metric has nonnegative scalar curvature, then it must be Yamabe positive. 

Corollary \ref{NonCor} tells us that there are many sets of AE seed data for which the conformal constraint equations admit no solutions. Labeling one such set as $\hat {\mathcal{S}}:=(\Sigma^n, \hat \lambda, \hat N, \hat \sigma,\hat \tau, \hat r, \hat j)$,  we may consider a sequence 
$\mathcal{S_\ell}:= (\Sigma^n, \lambda_\ell, N_\ell, \sigma_\ell,  \tau_\ell, r_\ell, j_\ell)$ 
of seed data sets such that for each element of the sequence $\mathcal{S_\ell}$ there \emph{is} a solution $(\phi_\ell, W_\ell)$, and such that the sequence $\mathcal{S_\ell}$ converges to 
$\hat {\mathcal{S}}$. We may then ask what the behavior of the sequence of solutions  $(\phi_\ell, W_\ell)$ might be. We first prove a result which shows that the solution  sequence cannot be bounded:

\begin{thm}[Unboundedness]  
\label{Unbdd}
Suppose that  $\hat {\mathcal{S}}:=(\Sigma^n, \hat \lambda, \hat N, \hat \sigma,\hat \tau, \hat r, \hat j)$ is a set of seed data satisfying the regularity conditions of Corollary \ref{Admiss}, and suppose that for this data, the conformal constraint equations admit no solution. Suppose that $\mathcal{S_\ell}:= (\Sigma^n, \lambda_\ell, N_\ell, \sigma_\ell,  \tau_\ell, r_\ell, j_\ell)$ is a sequence of seed data sets such that each element $\mathcal{S_\ell}$ of the sequence satisfies the regularity conditions of Corollary \ref{Admiss} and for each element the conformal constraint equations admits a solution $(\phi_\ell, W_\ell)$, and suppose as well that $\mathcal{S_\ell}$ converges in $W^{1,p}_{\delta-1}$ to $\hat {\mathcal{S}}$. There do not exist constants $a$ and $b$ such that 
\begin{equation} 
\label{lims}
0<a\leq \phi_\ell \leq b
\end{equation} 
for all $\ell$.
\end{thm}

\begin{proof}
Setting up a proof by contradiction, we presume that such constants $a$ and $b$ do exist. It follows that for the sequence of seed data $\mathcal{S_\ell}$, the right hand side of \eqref{origVect} is uniformly bounded, and contained in $L^p_{\delta-2}$. Thence, since  a $W^{2,p}_\delta$-AE manifold does not admit any conformal Killing fields, we determine from Lemma \ref{VectLap} that the vector fields $W_\ell$ solving \eqref{origVect} are uniformly bounded and contained in $W^{2,p}_\delta.$ 

We now focus on the Lichnerowicz equation \eqref{origLich}, which for convenience we write in the form 
$\alpha_n \Delta_\ell \phi_\ell = F_\ell(W_\ell, \phi_\ell).$ Combining the presumed bounds \eqref{lims} on $\phi_\ell$ with the uniform bounds on $W_\ell$ obtained above, along with the hypothesized regularity of the seed data $\mathcal{S_\ell}$, we see that $ F_\ell(W_\ell, \phi_\ell)$ is uniformly bounded in $L^p_{\delta-2}.$
Consequently, the solutions $\phi_\ell$ of the Lichnerowicz equation are uniformly bounded in $W^{2,p}_\delta.$ Since the embedding of $W^{2,p}_\delta$ in $L^\infty$ is compact, the sequence $\phi_\ell$ must contain a subsequence $\tilde \phi_m$ which converges (in $L^\infty$) to some positive function $\phi_\infty$. By a similar argument, the corresponding subsequence $\tilde W_m$ converges to a vector field $W_\infty$ in $W^{2,p}_\delta$. 

We now apply standard elliptic bootstrap techniques to the converging sequence $(\tilde W_m, \tilde \phi_m) \rightarrow (W_\infty, \phi_\infty)$ to argue that $(W_\infty, \phi_\infty)$ must be a solution to the conformal constraint equations for the limiting seed data set $\hat {\mathcal {S}}$. However, by assumption Eqns.\ \eqref{origLich}-\eqref{origVect} do not admit a solution for the seed data $\hat {\mathcal {S}}$. We thus obtain the contradiction which shows that in fact the sequence $(W_\ell, \phi_\ell)$ cannot be bounded away from both zero and infinity.

\end{proof}

This result does not tell us whether, in general, the sequence $\phi_m$ blows up or goes to zero. To obtain results which distinguish these possibilities, we make further assumptions.  Unlike Theorem  \ref{Unbdd}, these further results (below) are somewhat restrictive regarding both the nature of the seed data sets $\hat {\mathcal {S}}$ for which (by assumption)  no solutions to the conformal constraint equations exist, and the nature of the sequences of seed data sets $\mathcal{S_\ell}$ which converge to $\hat {\mathcal {S}}$. These results  hold for the seed data sets of Corollary \ref{NonCor}, as well as for a wider class; however, it is not clear whether they hold for every possible  choice of $\hat {\mathcal {S}}$ and of $\mathcal{S_\ell}$. We hope to determine that sometime in the future. 

To go beyond Theorem \ref{Unbdd}, the following monotonicity lemma is very useful. We note that some  of the restrictions on the choice of $\hat {\mathcal {S}}$ and of $\mathcal{S_\ell}$ originate here, in the hypothesis of this lemma. We note in particular that we need 
higher regularity (expressed here, for convenience, using H\"{o}lder norms) of the seed data in order to get the pointwise bounds we require.

\begin{lem} [Monotonicity]
\label{Monotonicity}
Suppose that $(\Sigma^n; \lambda_\ell)$ is a sequence of $C^{2,\alpha}_\delta$-AE geometries which converge (in $C^{2,\alpha}_\delta$) to $(\Sigma^n; \lambda_\infty)$, and that $N_\ell-1\in C^{2,\alpha}_\delta$ similarly converges. Suppose that $\tau_\ell^2$ is a sequence of $C^{0,\alpha}_{\delta-2}$ functions which converge (in this space) to $\tau_\infty^2$, suppose that $f_\ell^2$ and $r_\ell$ are sequences of $C^{0,\alpha}_{\delta-2}$ functions, and finally suppose that for each index $\ell$, the Lichnerowicz equation \eqref{LichneroEqn} corresponding to the data $(\Sigma^n, \lambda_\ell, N_\ell, f_\ell, \tau_\ell, r_\ell)$ admits a solution $\phi_\ell$. Let the function $\tau_0$ be defined as
\begin{equation}
\label{tau0}
\tau_0^2:=C \rho^{\delta-2},
\end{equation}
where $C$ is a positive constant sufficiently large so that $\tau_0^2 \geq \tau_\ell^2$ and $\kappa_n \tau_0^2 \geq -R_{\lambda_\ell}$ for all $\ell$. If we label as $\psi_\infty$  the  conformal factor for which $R_{\psi_\infty^{q_n-2} \lambda_\infty} = -\kappa_n  \tau_0^2$, then for any $\epsilon>0$ one has $\phi_\ell > \psi_\infty - \epsilon$ for sufficiently large $\ell$.
\end{lem}

We note that there is no assumption in this lemma regarding the Yamabe class of the metrics $\lambda_\ell$ and $\lambda_\infty$; nor is there any assumption that the sequences $f_\ell$ or $r_\ell$ converge. We also note that the existence of $C$ follows from the function space conditions placed on $\lambda_\ell$ and $\tau_\ell$, and the existence of $\psi_\infty$ follows from the presumed form \eqref{tau0} of $\tau_0$, and from the properties of the AE Yamabe classes, as described in Lemma \ref{Yam&CurvCrit}.

\begin{proof} 
The hypothesis of this lemma presumes that for each set of the data $(\Sigma^n, \lambda_\ell, N_\ell, f_\ell, \tau_\ell, r_\ell)$, a positive solution $\phi_\ell$ of the corresponding Lichnerowicz equation exists. The form  \eqref{tau0} of $\tau_0$ together with  Lemma \ref{Yam&CurvCrit} imply 
 that for each value of $\ell$ there exists a conformal function $\psi_\ell$ for which $R_{\psi_\ell^{q_n -2} \lambda_\ell} = -\kappa_n  \tau_0^2$. Noting that the functions $\psi_\ell$ and $\phi_\ell$ are all expected to approach one (or some other constant) asymptotically, we seek to show here that $\phi_\ell \geq\psi_\ell$ for all $\ell$. 

It follows from the definitions of $\phi_\ell$ and $\psi_\ell$ and from the conformal covariance of the Lichnerowicz equation under the conformal transformation
\begin{equation}
(\lambda, f, \tau, r) \rightarrow (\theta^{q_n -2} \lambda, \theta^{-q_n } f, \theta^{-\frac32 q_n +1} r)
\end{equation} 
that if we set $\tilde \varphi_\ell:= \frac{\phi_\ell}{\psi_\ell}$, then $\tilde \varphi_\ell$ satisfies
\begin{equation}
\label{varphi}
\alpha_n\Delta_{\psi_\ell^{q_n -2} \lambda_\ell} \tilde\varphi_\ell = -\kappa_n  \tau_0^2 \tilde \varphi_\ell + \kappa_n  \tau_\ell^2 \tilde \varphi_\ell^{q_n -1} -f^2 \tilde \varphi_\ell^{-q_n -1} -r \tilde \varphi_\ell ^{-\frac{q_n }{2}}.
\end{equation}

To obtain a contradiction with our contention that $\phi_\ell \geq \psi_\ell$ for all $\ell$, we suppose now that $\phi_\ell < \psi_\ell$ somewhere in $\Sigma^n$ for some $\ell$. This implies that $\tilde \varphi_\ell<1$ somewhere. Since $\tilde \varphi_\ell$ approaches $1$  asymptotically, this function must have a global minimum at some point $p\in \Sigma^n$; hence (since $\tilde \varphi_\ell$ is, by construction, continuous) there exists a small ball $B(p) \subset \Sigma^n$ containing $p$ on which $\tilde \varphi_\ell<1$. It immediately follows that $-\kappa_n  \tau_0^2 \tilde \varphi_\ell + \kappa_n  \tau_\ell^2 \tilde \varphi_\ell^{q_n -1}<0$ on $B(p)$. Combining this with \eqref{varphi}, we find that  $\Delta_{\psi_\ell^{q_n -2} \lambda_\ell} \tilde\varphi_\ell \leq 0$ on $B(p)$. Since a minimum is achieved in the interior of $B(p)$, it follows from the maximum principle that $\tilde \varphi_\ell$ is constant (and negative) on $B(p)$. It follows now from standard arguments that the constancy of $\tilde \varphi_\ell$ on $B(p)$ extends to all of $\Sigma^n$. Noting the asymptotic behavior of $\tilde \varphi_\ell$, we obtain a contradiction; consequently, $\phi_\ell \geq \psi_\ell$ for all $\ell$, everywhere on $\Sigma^n$.

Having established this inequality, to complete the proof of this lemma, it is sufficient to show that $\psi_\ell \rightarrow \psi_\infty$ in $C^{2,\alpha}_\delta$. As a step towards verifying  this limit, we first show that for all $\ell$, $\psi_\ell \leq1$  everywhere. To verify this, suppose that $\psi_\ell>1$ somewhere. It follows from the asymptotic behavior of $\psi_\ell$ that there is a point $q$ at which $\psi_\ell$ achieves a maximum. The regularity of $\psi_\ell$ together with the conformal transformation equation for scalar curvature and the definition of $\psi_\ell$ now imply that  $\psi_\ell^{q_n -2} \leq \dfrac{-R_{\lambda_\ell}}{\kappa_n \tau_0^2}$ at $q$. This violates the hypothesized inequality relating $R_{\lambda_\ell}$ and $\tau_0$; we therefore conclude that $\psi_\ell \leq1$.

We now use this boundedness of $\psi_\ell$ to argue the convergence of this sequence. Recall that, by definition, the functions $\psi_\ell$ satisfy
\begin{equation}
\label{psieq}
 \alpha_n\Delta_{\lambda_\ell} \psi_\ell  = R_{\lambda_\ell} \psi_\ell +\kappa_n  \tau_0^2 \psi_\ell ^{q_n  -1}. 
 \end{equation}
The presumed regularity of the sequence of metrics $\lambda_\ell$ and of the function $\tau_0$, together with the bounds on $\psi_\ell$, allow us to use \eqref{psieq} to bootstrap the regularity of $\psi_\ell$ so that $\psi_\ell-1 \in C^{2,\alpha}_\delta$. We may then use compact embedding to show that there exists a subsequence $\hat \psi_m$ such that $\hat \psi_m-1$ converges in $C^{1,\alpha}_{\delta'}$ for some $\delta'>\delta$ to a function $\hat \psi_\infty-1$. 

A priori, we do not know that $\hat \psi_\infty$ is the conformal factor for which $R_{\hat \psi_\infty^{q_n -2} \lambda_\infty}= -\kappa_n  \tau_0^2$. To argue that it is, we add identical terms to both sides of \eqref{psieq}, rearrange terms, and obtain
\begin{multline}
\label{hatpsi}
 (-\alpha_n \Delta_{\lambda_\infty}+R_{\lambda_\infty})\hat \psi_m \\ =
 [(-\alpha_n\Delta_{\lambda_\infty}+R_{\lambda_\infty})-(-\alpha_n\Delta_{\lambda_m}+R_{\lambda_m})]\hat \psi_m -\kappa_n \tau_0^2 \hat \psi_m^{q_n -1}.
 \end{multline}
Since the sequence of metrics $\lambda_m$ converges in $C^{2,\alpha}_\delta$, the sequence of operators 
$-\alpha_n\Delta_{\lambda_m}+R_{\lambda_m}$ does as well. Combining this with the convergence of $\hat \psi_m$, we see that the sequence of terms  $[(-\alpha_n\Delta_{\lambda_\infty}+R_{\lambda_\infty})-(-\alpha_n\Delta_{\lambda_m}+R_{\lambda_m})]\hat \psi_m$ in \eqref{hatpsi} converges to zero. The convergence of  $\hat \psi_m$ guarantees that the remaining term on the right hand side of \eqref{hatpsi} converges; it then follows from the uniqueness Lemma \ref{uniqueness} that $\hat \psi_\infty$ is indeed the conformal factor for which $R_{\hat \psi_\infty \lambda_\infty}= -\kappa_n  \tau_0^2$, thus completing the proof of this lemma. 
\end{proof} 

We now combine Lemma \ref{Monotonicity} with Theorem \ref{Unbdd} to obtain a blow-up result for sequences of solutions of the conformal constraint equations. 

\begin{thm} [A Blow-up Result]
\label{blowup}
Suppose that $(\Sigma^n; \lambda_\ell, N_\ell, \sigma_\ell, \tau_\ell, r_\ell, j_\ell)$ is a sequence of $C^{2,\alpha}_\delta$ asymptotically Euclidean seed data with $\tau_\ell \in C^{1,\alpha}_{\delta-1}$, and with $\sigma_\ell, r_\ell, j_\ell \in C^{0,\alpha}_{\delta-1}$, for $\delta \in (2-n, 0)$. Suppose that the conformal constraint equations admit a solution $(\varphi_\ell,  W_\ell)$ for each $\ell$, and finally suppose that the sequence $(\Sigma^n; \lambda_\ell, N_\ell, \sigma_\ell, \tau_\ell, r_\ell, j_\ell)$ converges uniformly (in the indicated spaces) to a set of asymptotically Euclidean seed data $(\Sigma^n; \hat \lambda, \hat N, \hat  \sigma, \hat \tau, \hat r, \hat j)$ for which the conformal constraint equations admit no solution. Then $\sup \varphi_\ell \rightarrow \infty$.
\end{thm}

\begin{proof}
Since we know that each of the functions $\tau_\ell$ asymptotically approaches zero, and since we know that the sequence of functions $\tau_\ell$ has a bounded limit, it follows that we can choose a sufficiently large positive constant $C$ so that $\tau_\ell \leq C\rho^{\delta -1}$ for a uniform radial function $\rho$. If we now set $f_\ell^2 = |\sigma_\ell + \frac{1}{2N} L_{\lambda_\ell}W|^2_{\lambda_\ell}$, then the hypothesis of Lemma \ref{Monotonicity} is satisfied. It thus follows that $\phi_\ell$ is bounded away from zero. We readily verify that the hypothesis of Theorem \ref{Unbdd} is also satisfied. Consequently we know that $\varphi_\ell$ cannot be bounded away from both $0$ and $\infty$. Since we \emph{do} have $\phi_\ell$ bounded away from $0$, it follows that this sequence blows up.
\end{proof}

If we further tighten the conditions on the sequence of seed data sets in certain ways
then we can obtain further information on the blow up of the solutions. We present here one version of such a result; it is likely that this result could be generalized.

\begin{thm} [Another Blow-up Result]
\label{blowup2}
Suppose that the sequence of seed data sets $(\Sigma^n; \lambda_\ell, N_\ell, \sigma_\ell, \tau_\ell, r_\ell, j_\ell)$ satisfies all the conditions of Theorem \ref{blowup}, along with the additional conditions that $\tau_\ell \geq 0$, that $\tau_\ell \geq \tau_{\ell +1}$, that the limit mean curvature function $\hat \tau=0$, and that the limit metric $\hat \lambda$ is either in the Yamabe negative or the Yamabe zero class. Let $(\varphi_\ell, W_\ell)$ denote the corresponding solutions of the conformal constraint equations. For any choice of $p>n$, one (or both) of the following is true:
\begin{itemize}
\item  $\|\tau_\ell^2 \phi_\ell^{q_n -1}\|_{L^p_{\eta-2}}$ is unbounded, for  all $\eta \in (2-n,0)$.

\item  $\|R_{\lambda_\ell}\phi_\ell\|_{L^p_{\xi}}$ is unbounded, for all $\xi \in \R$.
\end{itemize}
\end{thm} 

\begin{proof}
Since we know by hypothesis that the conformal constraint equations  admit a solution for each of the sequence of seed data sets, it follows from Theorem \ref{LichIff} that there exists a sequence of conformal transformation mappings $\psi_\ell$ which map the scalar curvature to $-\kappa_n \tau_\ell^2$; they satisfy
\begin{equation}
\label{psieeq}
 \alpha_n\Delta_{\lambda_\ell} \psi_\ell  = R_{\lambda_\ell} \psi_\ell +\kappa_n  \tau_\ell^2 \psi_\ell ^{q_n  -1}. 
 \end{equation}
 It follows from the standard elliptic estimates (from \cite{CBC81}) together with the regularity presumed for the metrics $\lambda_\ell$  that the solutions $\psi_\ell$ to \eqref{psieeq} satisfy the estimate
 \begin{equation}
 \label{ineq}
 \|\psi_\ell-1\|_{W^{2,p}_\eta} \leq C\|\tau_\ell^2 \psi_\ell^{q_n -1}\|_{L^{p}_{\eta-2}}
                + C \|R_{\lambda_\ell}\psi_\ell\|_{L^p_{\xi}}
 \end{equation}
 for some constant $C$ and for any choices of $\eta \in (2-n,0)$ and of $\xi \in \R$. We note that the presumed regularity and convergence of the sequence of metrics $\lambda_\ell$ allows us to choose a single constant $C$, independent of $\ell$. 
 
We now argue that the right hand side of the inequality \eqref{ineq} is unbounded. We presume that this is not the case, and we seek a contradiction.  Since this presumption implies the uniform boundedness of $\psi_\ell-1$ in $W^{2,p}_\eta$, compact embedding implies the existence of a $C^{1,\alpha}_\eta$ converging subsequence $\psi_m$, with limit $\psi_\infty$. The boundedness of $\psi_\ell$, together with our hypothesis regarding the sequence $\tau_\ell$, also implies that $\tau^2_\ell \psi^{q_n -1}_\ell \rightarrow 0$. We may then use bootstrapping arguments to show that the limit function $\psi_\infty$ satisfies  $\alpha_n\Delta_{\lambda_\infty} \tilde \psi_\infty  = R_{\lambda_\infty} \tilde \psi_\infty$. However, since $\lambda_\infty$ is Yamabe nonpositive, there cannot be a solution to this equation. We thus obtain a contradiction, and consequently determine that either $\|\tau_\ell^2 \psi_\ell^{q_n -1}\|_{L^{p}_{\eta-2}}$ or $\|R_{\lambda_\ell} \psi_\ell\|_{L^p_{\xi}}$ is unbounded. 

To argue that either $\|\tau_\ell^2 \varphi_\ell^{q_n -1}\|_{L^{p}_{\eta-2}}$ or $\|R_{\lambda_\ell}\varphi_\ell\|_{L^p_{\xi}}$ (or both) is unbounded, we use the monotonicity of the $\tau_\ell$ sequence, together with arguments  similar to those used in proving Lemma \ref{Monotonicity}, to show that $\varphi_{\ell} \geq \psi_\ell$ for all $\ell$. The result follows.

 \end{proof} 
 
 \section{Existence Result for AE Seed Data with $\tau$ Admitting Zeroes}
 \label{zeros}

As noted above, for the class of asymptotically Euclidean seed data sets with $\tau =0$, the criterion for the existence of solutions to the conformal constraint equations is simple: Solutions exist if and only if  the metric is Yamabe positive. For AE seed data with nonconstant $\tau$, the few existence results known \cite{CBIY00, DIMM14} all involve positive Yamabe metrics as well, and also require that $\tau$ have no zeroes.\footnote{The conditions for asymptotically Euclidean seed data require $\tau$ to approach zero asymptotically. For the results cited here, $\tau$ approaches  zero but never crosses zero.} 

Here we present an existence theorem for solutions of the conformal constraint equations for seed data sets which include metrics that need not be Yamabe positive, and for choices of nonconstant $\tau$ which may admit zeroes. 
As with many existence theorems for these equations, the key to the proof is showing that there exist global sub and supersolutions (see Section \ref{AE}) for the system \eqref{origLich}-\eqref{origVect}, and the key to finding these involves balancing the positive and negative terms which appear in the Lichnerowicz equation \eqref{origLich}. The reason most results to date require $\tau$ to be nonzero is because the only terms appearing on the right hand side of Eqn.\ \eqref{origLich} which may be positive are the $\tau$ term and the $R_\lambda$ term; hence, setting $\tau^2 >0$ can balance the negative contributions from the $|\sigma + \frac{1}{2N}L_\lambda W|$ term and the $r$ term. We get around this requirement here by using curvature deformation results based on those appearing in \cite{DM15}. We state the deformation result we need in Lemma \ref{Deform} after stating a definition which is needed for this lemma; we then present the existence theorem below. 

\begin{defn}[Yamabe Invariant of a Subset]
\label{defn}
Let $(\Sigma^n; \lambda)$ be a Riemannian manifold, let $S$ be a measurable subset of $\Sigma^n$, and let $\mathcal{F}_S$ be the set of real valued functions (not identically zero, of sufficient regularity) on $\Sigma^n$ which vanish on the complementary set $\Sigma^n\setminus S$. The Yamabe invariant for $S$ is given by 
\begin{equation}
Y(S\subset \Sigma^n):= \inf_{u\in \mathcal{F}_S} \frac{\int_{\Sigma^n} \alpha_n|\nabla u|^2 +R_\lambda u^2}{\|u\|^2_{L^q_n }}.
\end{equation}
The set $S$ is labeled Yamabe positive, Yamabe negative, or Yamabe zero according to the sign of $Y(S\subset \Sigma^n)$.
\end{defn}

\begin{lem} [Curvature Deformation Lemma]
\label{Deform}
Let $(\Sigma^n; \lambda)$ be $W^{2,p}_\delta$-AE with $p>n$, with $\delta \in (2-n,0)$, and with radial function $\rho$, and let $S$ be a closed subset of $\Sigma^n$ which is Yamabe positive in the sense of Definition \ref{defn}. There exists a conformal factor $\Psi$, with $\Psi -1 \in W^{2,p}_\delta$, such that $R_{\Psi^{q_n -2} \lambda} \geq \epsilon \rho ^{\delta-2}$ on $S$ for some $\epsilon >0$, and such that $R_{\Psi^{q_n -2} \lambda} \geq -\zeta$ everywhere on $\Sigma^n$ for some constant $\zeta>0$.

Further,  if $S'$ is a (closed) subset of $S$, then there exists a conformal factor $\Psi'$ for the set $S'$, satisfying the corresponding conditions (as above) relative to $S'$, and also satisfying the inequalities $c \leq \Psi' \leq \Psi$, $\epsilon'\geq \epsilon$ and $\zeta'\leq \zeta$ for some positive constant $c$ depending only on the metric.
%such that there exists a constant $c>0$ dependent \emph{only} on the metric $\lambda$, with $c \leq \Psi' \leq \Psi$. In particular, $\epsilon'\geq \epsilon$ and $\zeta'\leq \sup \Psi^{q_n -1}$.
\end{lem}

\begin{proof}
The proof of this lemma depends to a large extent on results proven in \cite{DM15}. We  define a function $\mathcal D:\Sigma^n \rightarrow \R$ via 
\[\mathcal D(p):= \frac{2}{\pi}\arctan{ (\mathop{\mathrm{Distance}}_\lambda (p, S))}\leq 1,\] where the upper bound indicates a choice of branch.
It follows immediately from this definition that  $\{p\in \Sigma^n | \mathcal D(p)=0\} = S$. It then follows from the prescribed scalar curvature result Theorem 4.1 in \cite{DM15} that there exists a conformal transformation function $\Theta$ such that $R_{\Theta^{q_n -2} \lambda} = -\kappa_n \mathcal D^2\rho^{\delta-2}$.  In turn, we may now apply the Curvature Criterion Theorem 
\ref{LichIff} with $f^2 =\rho^{\delta-2}$ and $r=0$ and thereby verify that there exists a solution $\Psi$ to the Lichnerowicz equation
\begin{equation}
\label{Psieq}
0=\mathcal L(\Psi) := -\alpha_n \Delta_\lambda  \Psi + R_\lambda \Psi + \kappa_n  \mathcal D^2 \rho^{\delta-2} \Psi^{q_n -1} -\rho^{\delta-2} \Psi^{-q_n -1}.
\end{equation}
 
We claim that this function $\Psi$ satisfies the criteria stated in this Lemma. To verify this, we note that it follows from Eqn. \eqref{Psieq} that $R_{\Psi^{q_n -2} \lambda}= -\kappa_n  \mathcal D^2 \rho^{\delta -2} +\rho^{\delta-2} \Psi^{-2q_n }.$ The regularity and boundedness properties built into the definition of $\mathcal D$ show that $R_{\Psi^{q_n -2} \lambda}$ is bounded from below everywhere on $\Sigma^n$. The fact that $\mathcal D$ vanishes on $S$, together with the regularity and boundedness of $\Psi$ (a solution of \eqref{Psieq}) on the closed set $S$, show that there exists some $\epsilon>0$ such that $R_{\Psi^{q_n -2} \lambda} \geq \epsilon \rho ^{\delta-2}$ on $S$.

To prove the second statement, regarding the subset $S'$, we first define $\mathcal D'(p):= \frac{2}{\pi}\arctan{ (\mathop{\mathrm{Distance}}_\lambda (p, S'))}$,  and we see immediately that $\mathcal D'^2 \geq \mathcal D^2$. Hence, constructing first $\Theta'$ and then $\Psi'$ analogously to $\Theta$ and $\Psi$, we determine that $\Psi$ satisfies the supersolution inequality for $\Psi'$. Indeed, constructing the Lichnerowicz operator $\mathcal L'$ which corresponds to $S'$ and $\mathcal D'$ (and for which we have $\mathcal L'(\Psi')=0$), we calculate (using \eqref{Psieq}, and using the positivity of $\Psi$) 
\begin{equation}
\mathcal L'(\Psi) = \kappa_n \rho^{\delta-2}\Psi^{q_n-1} (\mathcal D'^2 - \mathcal D^2) \geq 0.
\end{equation}

This does not (directly) guarantee that $\Psi' \leq \Psi$. However, since we readily verify that for any positive value of $t\leq1$,  the quantity $t\Psi'$ satisfies the subsolution inequality for $\Psi'$, and since the boundedness of $\Psi$ and $\Psi'$ guarantee that there exist some positive $t_0$ such that $t_0 \Psi'\leq \Psi$, we see that indeed $t_0 \Psi'$ and $\Psi$ form a sub and supersolution pair for $\Psi'$. It then follows from Lemma \ref{SubSup} that $\Psi' \leq \Psi$. We note that in completing this argument, we use the Lichnerowicz solution uniqueness result  Proposition \ref{uniqueness}. 

To verify the positive lower bound for $\Psi'$ (also part of the second statement), we may use a variant of the argument implemented to prove Lemma \ref{Monotonicity}, since, again, $1 \geq \mathcal D'^2 \geq \mathcal D^2$. To verify the inequality for $\epsilon'$, we recall that $R_{\Psi'^{q_n -2} \lambda}= -\kappa_n  \mathcal D'^2 \rho^{\delta -2} +\rho^{\delta-2} \Psi'^{-2q_n }$, and we apply the bounds on $\Psi'$. For the inequality for $\zeta'$, we instead rely on the condition that $\mathcal D'^2 \leq 1$ and that $\rho\geq 1$.
%and, as above, construct $\Psi'$. Since $\zeta'(p)$ is bounded above for any given set, we can use a variation of the argument in Lemma \ref{Monotonicity} to show the lower bound. Next, note that since $\zeta'^2 \geq \zeta^2$, $\Psi$ is a supersolution for the equation for $\Psi'$. Also, it is straightforward to show that $t\Psi'$ is a subsolution for any $t\leq 1$. Picking $t$ small enough so that $t\Psi'<\Psi$, we use Lemma \ref{SubSup} to show that there is a solution between those two functions, and then Proposition \ref{uniqueness} to show that this solution is $\Psi'$. This gives the upper bound. The bounds on $\epsilon'$ and $\zeta'$ are then immediate.
\end{proof}

Our main result in this section is the following theorem. This is a near-CMC result, using an integral inequality on the derivative of $\tau$ as in, for example, \cite{HNT09}. 

\begin{thm} [Existence Theorem]
\label{Exist}
Suppose $(\Sigma^n; \lambda, N, \sigma, \tau, r, j)$ is a set of asymptotically Euclidean seed data which satisfies the regularity and admissibility conditions as stated in Corollary \ref{Admiss} and also satisfies $p>n$. Suppose in addition that there exists a positive constant $\alpha$ such that $S_\alpha:=\{p\in \Sigma^n | \kappa_n \tau^2(p	) \leq \alpha\}\subseteq S_0$ for some Yamabe positive set $S_0$. Then, there exists $M=M(\lambda,N, S_0)$ such that if $\alpha - M\|d\tau\|_{L^p_{\delta-2}}^2\geq 0$, and if $\sigma$, $r$ and $j$ are small enough (relative to $\lambda$, $N$, $S_0$, $\alpha$ and $\|d\tau\|_{L^p_{\delta-2}}^{-1}$) then there exist solutions to the conformal constraint equations of appropriate regularity.
\end{thm}

The existence of such an $\alpha$ is equivalent to the condition that the zero set of the function $\tau$ is sufficiently small. We include the set $S_0$ in the statement of this theorem  to emphasize the fact that  that the dependence of the constant $M$ on $\tau$ is very weak. In particular, $M$ depends only on a (Yamabe positive) bounding set $S_0$  containing the set $S_\alpha$. For example,  in considering a family of mean curvature functions $\tau_\ell$, as long as the corresponding sets $S_{\alpha_\ell}$ are nested, $M$ can be chosen uniformly. We use this fact in showing that there are seed data sets which satisfy the  hypothesis of this theorem. 

\begin{proof} 
As hypothesized, there exists a positive constant $\alpha$ such that the set $S_\alpha$ is Yamabe positive. It then follows from Lemma \ref{Deform} that we may choose a function 
$\Psi$ with $\Psi -1 \in W^{2,p}_\delta$ such that the scalar curvature $R_{\Psi^{q_n -2} \lambda} \geq \epsilon \rho ^{\delta-2}$ on $S_\alpha$, for some $\epsilon>0$.  We note that $R_{\Psi^{q_n -2} \lambda} \in L^p_{\delta -2}$ and that Lemma \ref{Deform} proves that the lower bound $-\zeta$ for $R_{\Psi^{q_n -2} \lambda}$, the value $\epsilon$, as well as the upper and lower bounds on $\Psi$, depend only on $\lambda$, $N$, and $S_0$ (or, on $S_\alpha$ if we take $S_0 = S_\alpha$). In particular, we use the conformal covariance of the CTS-H method (as explained in the introduction) and work with conformally transformed quantities, denoted by hats. We note that as a consequence of the upper and lower bounds on $\Psi$, 
bounds on hatted quantities are easily converted to bounds on the original seed data.

It follows from Theorem \ref{GlobalSubAndSup} that to prove that the conformal constraint equations admit a solution, it is sufficient to find a global sub and supersolution pair. We claim first that if $|\hat \sigma|$ and $\hat r$ (and consequently $\hat j$) are sufficiently small, then there exists a constant global supersolution $\eta$. To show this, we substitute $\eta$ into the inequality (involving the terms in \eqref{origLich}) which must be satisfied if this is the case. Doing a bit of rearranging, we see that $\eta$ is a global supersolution so long as the inequality 
\begin{equation} 
\label{superineq}
R_{\hat \lambda} \eta^{2-q_n} + \kappa_n  \tau^2- \left| \hat \sigma +\frac{1}{2\hat N}L_{\hat \lambda} W\right|^2 \eta^{-2q_n} - \hat r \eta^{\frac{2-3 q_n }{2}} \geq 0
\end{equation}
holds. We now work with the term involving $L_{\hat \lambda} W$, seeking to bound it from below for all allowable values of $W$. The standard quadratic inequality gives us $-| \hat \sigma +\frac{1}{2\hat N} L_{\hat \lambda} W|^2 \geq -2 |\hat \sigma|^2 -\frac{1}{\hat N} |L_{\hat \lambda} W|^2.$ Elliptic estimates based on Eqn.\ \eqref{origVect}, with $\varphi \leq \eta$ (the purported global supersolution) give us 
\begin{equation}
\|L_{\hat\lambda} W\|_{C^0_{\delta-1}} \leq c \|W\|_{W^{2,p}_{\delta}} \leq c\|\varphi ^{q_n} d\tau +\hat j\|_{L^p_{\delta-2}} \leq c \eta^{q_n} \|d\tau\|_{L^p_{\delta-2}} + c\|\hat j\|_{L^p_{\delta-2}},
\end{equation}
where $c$ is a constant that depends on the metric $\lambda$ and the lapse function $N$ only.
Then, since (following from the definition of the weighted norms) we have the pointwise estimate 
$|L_{\hat \lambda} W| \leq \| L_{\hat\lambda} W\|_{C^0_{\delta -1} }\rho ^{\delta -1}$, the needed inequality takes the form
\begin{equation}\label{Need}
R_{\hat\lambda} \eta^{2-{q_n}} + \kappa_n  \tau^2 -  (2|\hat\sigma|^2+c\|\hat j\|_{L^p_{\delta-2}})  \eta^{-2 {q_n}} - \hat r \eta^{\frac{2-3 {q_n}}{2}} \\
-c \|d\tau\|_{L^p_{\delta-2}} \rho^{2 \delta -2} \geq 0.
\end{equation}

We verify the inequality \eqref{Need} separately in the region $S_\alpha$, and in its complement. In $S_\alpha$, we have the scalar curvature bound $R_{\hat\lambda} \geq \epsilon \rho ^{\delta-2}$, for some fixed value of $\epsilon$. Hence, in $S_\alpha$, we may use the scalar curvature term to dominate the negative terms in \eqref{Need}. Specifically, if we choose $\eta$ so that 
\begin{equation}
\label{eta}
\eta^{2-{q_n}} = 2 c \|d\tau\|_{L^p_{\delta-2}} /\epsilon,
\end{equation}
then we verify that half of the scalar curvature term dominates the $d\tau$ term:
\begin{equation}
\label{halfS}
\frac{1}{2}R_{\hat \lambda} \eta^{2-{q_n}} - c \|d\tau\|_{L^p_{\delta-2}} \rho^{2 \delta -2} \geq \frac{1}{2} \epsilon \rho ^{\delta-2} \eta^{2-{q_n}} - c \|d\tau\|_{L^p_{\delta-2}} \rho^{2 \delta -2} \geq 0.
\end{equation}
We note that here, the choice of the radial function so that $\rho \geq 1$ is crucial; as well, we recall that $\delta$ is negative, by assumption. 

To take care of the rest of the negative terms in \eqref{Need} (still working on $S_\alpha$), we impose smallness conditions on $|\hat\sigma|$, on $|\hat j|$ and on $\hat r$. Specifically, with $\eta$ now fixed, we require $|\hat \sigma|$, $|\hat j|$ and $\hat r$ to be small enough so that 
\begin{equation}
\label{otherhalfS}
\frac{1}{2} \epsilon \rho ^{\delta-2} \eta ^{2-{q_n}}- (2|\hat \sigma|^2+c\|\hat j\|_{L^p_{\delta-2}}) \eta^{-2 {q_n}} - \hat r \eta^{\frac{2-3 {q_n}}{2}} \geq 0.
\end{equation}
Clearly these restrictions on the choice of the seed data can always be made.

We now determine which conditions on the seed data must be imposed in order to verify  inequality \eqref{Need} in the region $S^c_\alpha$ which is the complement of $S_\alpha$. To carry out this determination, we note that the following estimates hold within $S^c_\alpha$: i) the definition of $S^c_\alpha$ implies that $\kappa_n \tau^2>\alpha$; ii) Lemma \ref{Deform} guarantees that $R_{\hat \lambda} $ is bounded below by some (generally negative) constant, which we label $-\zeta$; iii) since, by definition, $\rho\geq 1$, the radial quantity $\rho^{2 \delta -2}\leq 1$. Combining these estimates with the specification  \eqref{eta} for the constant $\eta$, we can express the needed inequality \eqref{Need} (for the region $S^c_\alpha$) in the form
\begin{equation} 
\label{ScNeed}
-\zeta (2 c \|d\tau\|_{L^p_{\delta-2}} /\epsilon) +\alpha - (2|\hat\sigma|^2+c\|\hat j\|_{L^p_{\delta-2}})  \eta^{-2 {q_n}} - \hat r \eta^{\frac{2-3 {q_n}}{2}} -c \|d\tau\|_{L^p_{\delta-2}} \geq 0, 
\end{equation}
which can be rearranged into
\begin{equation}
\label{ScNeed}
\alpha - (2|\hat \sigma|^2+c\|\hat j\|_{L^p_{\delta-2}}) \eta^{-2 {q_n}} - \hat r \eta^{\frac{2-3 {q_n}}{2}} -\hat c \|d\tau\|_{L^p_{\delta-2}} \geq 0,
\end{equation} 
where the constant $\hat c := c( \frac {2 \zeta}{\epsilon} +1)$ depends only on the metric, $N$ and $S_0$. Splitting this inequality into a pair, we see that for a specified AE geometry $(\Sigma^n; \lambda)$  and a specified choice of $\alpha$ (recall that the constant $\alpha$ must be chosen so that $S_\alpha$ is Yamabe positive) it is sufficient to choose $\tau$ so that 
\begin{equation}
\label{M}
\|d\tau\|_{L^p_{\delta-2}} \leq \frac{\alpha}{2 \hat c },
\end{equation} 
and then choose $\hat \sigma$, $\hat r$ and $\hat j$ so that 
\begin{equation}
\label{small} 
(2|\hat\sigma|^2+c\|\hat j\|_{L^p_{\delta-2}}) \eta^{-2 {q_n}} + \hat r \eta^{\frac{2-3 {q_n}}{2}} \leq \frac{\alpha}{2}.
\end{equation}
We note that Eqn.\ \eqref{M} determines the constant  $M$  which appears in the hypothesis of this theorem.

We now have conditions on the seed data which guarantee that $\eta$ serves as a global supersolution for the system \eqref{origLich}-\eqref{origVect}. For a global subsolution we choose $\xi \psi/\Psi $ where $\psi$ is the conformal factor for which $R_{\psi^{{q_n} -2} \lambda} = - \kappa_n  \tau^2$ (the existence of such a function $\psi$ is guaranteed by the hypothesis that the seed data be admissible), $\Psi$ is as before, and where $\xi$ is a constant between zero and one which is chosen to ensure the sub/supersolution inequality $\xi \psi/\Psi \leq \eta$. (The division by $\Psi$ is to account for the conformal transformation that we already used.) One readily verifies that the appropriate differential inequality is satisfied so that indeed $\xi \psi/\Psi$ is a global subsolution.  We have thus determined that for any AE seed data satisfying the conditions \eqref{M}, \eqref{otherhalfS} and \eqref{small}, the conformal constraint equations admit a solution. 

\end{proof}

We wish to emphasize that there do exist seed data sets satisfying the hypothesis of Theorem \ref{Exist}.  Indeed, one may construct such data as follows: Choosing any asymptotically Euclidean geometry $(\Sigma^n, \lambda, N)$ of sufficient regularity, one considers smooth functions $\tau$ which are unity inside $B_{\rho_0}$, vanish outside $B_{2\rho_0}$, and have derivatives as small as consistently possible in the annulus $B_{2\rho_0} \setminus B_{\rho_0}$. One readily checks (see \cite{DM15}) that for large enough $\rho_0$, Lemma \ref{Deform} holds on $S:= \Sigma^n \setminus B_{\rho_0}$. One also readily checks that $\|d\tau\|_{L^p_{\delta-2}}$ can be made arbitrarily small by choosing large enough $\rho_0$ (as long as $\delta>-1$). Since the zero sets of functions $\tau$ constructed in this way are strictly decreasing as $\rho_0$ increases, it follows from the estimates stated in the
second part of Lemma \ref{Deform} that the constant $M$ is uniformly bounded, and so the condition \eqref{M} is satisfied for sufficiently large $\rho_0$. Conditions \eqref{otherhalfS} and \eqref{small} are met by choosing small $|\sigma|$, $r$ and $|j|$ directly.

It should be mentioned that there is no evidence of uniqueness for this result, except that it is near-CMC in some sense. The sub and supersolution theorem only show existence, and never uniqueness. Indeed, we expect that for some seed data, there are multiple solutions to the conformal constraint equations, as observed, for instance, in \cite{Nguyen15}. 

\section{Acknowledgments}
This work was partially supported by the National Science Foundation through grant DMS-63431. We thank the Mathematical Sciences Research Institute and the Department of Physics at the University of Maryland, where some of this work was carried out.

\bibliographystyle{alpha}
\bibliography{bibliography}
\end{document}